\documentclass[
aps,
pra,
reprint,
superscriptaddress
]{revtex4-2}

\usepackage[T1]{fontenc}
\usepackage[utf8]{inputenc}
\usepackage{lmodern}
\usepackage{mathtools}
\usepackage{amssymb}
\usepackage{amsthm}
\usepackage{amsmath}
\usepackage{amssymb}
\usepackage{bm}
\usepackage{graphicx}
\usepackage{dcolumn}
\usepackage{xcolor}
\usepackage{siunitx}
\usepackage[
colorlinks=true,
linktocpage=true
]{hyperref}
\usepackage{longtable}

\makeatletter
\newtheorem*{rep@theorem}{\rep@title}
\newcommand{\newreptheorem}[2]{%
\newenvironment{rep#1}[1]{%
 \def\rep@title{#2 \ref{##1}}%
 \begin{rep@theorem}}%
 {\end{rep@theorem}}}
\makeatother

\newtheorem{theorem}{Theorem}
\newreptheorem{theorem}{Theorem}
\newtheorem{lemma}[theorem]{Lemma}
\newtheorem{corollary}[theorem]{Corollary}

\newtheorem{remark}[theorem]{Remark}

\newtheorem*{def:Tri_Game_6}{$\boldsymbol{C_6}$ hidden linear function game}
\newtheorem*{def:C_6_Stabilizer_Game}{$\boldsymbol{C_6}$ cubic Boolean function games}
\newtheorem*{def:Stab_Sub_Game}{$\boldsymbol{C_6}$ stabilizer submeasurement games}
\newtheorem*{def:C_n_Stabilizer_Game}{$\boldsymbol{C_n}$ cubic Boolean function games}
\newtheorem*{def:general_Stab_Sub_Game}{$\boldsymbol{C_n}$ stabilizer submeasurement games}
\newtheorem{thm:Depth_D_Separation}{Theorem}
\newtheorem*{lem:C6_Stabilizers}{Lemma}

\newcommand{\R}{\mathbb{R}}

\DeclareMathOperator{\V}{Var}
\DeclareMathOperator{\Cov}{Cov}
\DeclareMathOperator{\tr}{tr}
\newcommand{\ket}[1]{| #1 \rangle}
\newcommand{\bra}[1]{\langle #1 |}

\newcommand{\op}[2]{\ket{#1}\!\bra{#2}}

\newcommand{\ev}[1]{\langle #1 \rangle}
\newcommand{\T}{\mathsf{T}}
\renewcommand{\l}[1]{\mathopen{}\left#1}
\renewcommand{\r}[1]{\right#1\mathclose{}}

\renewcommand{\S}{\mathcal{S}}
\newcommand{\fid}{\mathcal{F}}
\DeclareMathOperator{\sgn}{\mathrm{sgn}}
\newcommand{\CZ}{\mathrm{C}Z}

\usepackage{suffix}

\newcommand{\az}[1]{{\textcolor{teal}{[#1]}}}
\WithSuffix\newcommand\az*[1]{{\textcolor{teal}{#1}}}

\begin{document}

\title{Quantum computational advantage attested by nonlocal games with the cyclic cluster state}

\author{Austin K. Daniel}\email{austindaniel@unm.edu}
\affiliation{Department of Physics and Astronomy, Center for Quantum Information and Control, University of New Mexico, Albuquerque, New Mexico 87106, USA}

\author{Yingyue Zhu}
\affiliation{Joint Quantum Institute, Department of Physics, University of Maryland, College Park, MD 20742, USA}

\author{C. Huerta Alderete}
\affiliation{Joint Quantum Institute, Department of Physics, University of Maryland, College Park, MD 20742, USA}

\author{Vikas Buchemmavari}
\affiliation{Department of Physics and Astronomy, Center for Quantum Information and Control, University of New Mexico, Albuquerque, New Mexico 87106, USA}

\author{Alaina M. Green}
\affiliation{Joint Quantum Institute, Department of Physics, University of Maryland, College Park, MD 20742, USA}

\author{Nhung H. Nguyen}
\affiliation{Joint Quantum Institute, Department of Physics, University of Maryland, College Park, MD 20742, USA}

\author{Tyler G. Thurtell}
\affiliation{Department of Physics and Astronomy, Center for Quantum Information and Control, University of New Mexico, Albuquerque, New Mexico 87106, USA}

\author{Andrew Zhao}
\affiliation{Department of Physics and Astronomy, Center for Quantum Information and Control, University of New Mexico, Albuquerque, New Mexico 87106, USA}

\author{Norbert M. Linke}\email{linke@umd.edu}
\affiliation{Joint Quantum Institute, Department of Physics, University of Maryland, College Park, MD 20742, USA}

\author{Akimasa Miyake}\email{amiyake@unm.edu}
\affiliation{Department of Physics and Astronomy, Center for Quantum Information and Control, University of New Mexico, Albuquerque, New Mexico 87106, USA}

\date{\today}


\begin{abstract}

    We propose a set of Bell-type nonlocal games  {that can be used to prove an unconditional quantum advantage in an objective and hardware-agnostic manner}. In these games, the circuit depth needed to prepare a cyclic cluster state and measure a subset of its Pauli stabilizers on a quantum computer is compared to that of classical Boolean circuits with the same, nearest-neighboring gate connectivity. 
    Using a circuit-based trapped-ion quantum computer, we prepare and measure a six-qubit cyclic cluster state with an overall fidelity of 60.6\% and 66.4\%, before and after correcting for measurement-readout errors, respectively. Our experimental results indicate that while this fidelity readily passes conventional (or depth-0) Bell bounds for local hidden-variable models, it is on the cusp of demonstrating  {a higher probability of success than what is possible by} depth-1 classical circuits. Our games offer a practical and scalable set of quantitative benchmarks for quantum computers in the pre-fault-tolerant regime as the number of qubits available increases.

\end{abstract}

\maketitle



\section{\label{sec:intro}Introduction}

There are many metrics to characterize the quality of noisy intermediate-scale quantum (NISQ) computers~\cite{preskill2018quantum}. For example, qubit count, gate count, gate fidelities, and quantum volume~\cite{cross2019validating} are common options. However, it is generally agreed upon that a more comprehensive picture is given by the device's overall performance in executing a variety of computational tasks. To this end, we look toward computational tasks with two desirable properties in this work. First, they should have objective targets, beyond which one can prove that the NISQ computer has outperformed some particular model of classical computation. This is different from the task of demonstrating so-called quantum supremacy via sampling random quantum circuits~\cite{arute2019quantum,wu2021strong} or bosonic linear interferometers~\cite{zhong2020quantum}, which relies on assumptions of the underlying problem's computational hardness, making it a moving target based on current state-of-the-art classical hardware and algorithms. In contrast, we seek an unconditional demonstration of quantum computational advantage. Second, the computational task should be agnostic to the choice of hardware implementation, so that it allows a fair comparison of results among different architectures. In particular, as most of NISQ computers have only geometrically-local entangling gates, here we consider a one-dimensional (1D) geometry with a periodic boundary condition (i.e., a cyclic array of qubits). 

Natural candidates satisfying the first desideratum are Bell-type nonlocal games~\cite{bell1964on,cleve2004consequences,Brunner_review_nonlocality_2014}, which have a long history of experimental demonstrations (see Refs.~\cite[Sec.~VII]{Brunner_review_nonlocality_2014} and \cite{hensen2015loophole, shalm2015strong, giustina2015significant, rosenfeld2017event}) and
have found renewed interest in the context of classically verifiable quantum advantage~\cite{brakersi2018cryptographic,brakerski2020simpler, kahanamoku2021classically}.
The violation of a Bell inequality sets an objective threshold that a quantum device must surpass in order to evade description by an analogous classical model. Extending these results to more general causal scenarios has been a topic of recent interest~\cite{fritz2012beyond, chaves2015unifying, chaves2017causal}.
 {For example, Bell scenarios based on generalized Greenberger-Horne-Zeiliinger (GHZ) states, an example of a so-called graph state~\cite{briegel2001persistent, hein2006entanglement, scarani2005nonlocality, guhne2005bell, waegell2019benchmarks}, have been shown to be capable of computing arbitrary Boolean functions~\cite{hoban2011non-adaptive, Demirel_2021}.  Yet, preparing such a state with nearest-neighboring entangling gates on a 1D geometry requires linear-depth circuits.  Thus we don't expect any quantum advantage when comparing quantum and classical circuit depths that accomplish this task.}

 {On the other hand}, cyclic cluster states are graph states that require only nearest-neighbor entangling gates. This restricted connectivity is inherent to a wide variety of quantum-computing platforms, thus satisfying our second desideratum. Cyclic cluster states also form the basis of Bell-type scenarios that go beyond the traditional locality assumptions of Bell's theorem, refuting even classical theories assisted by a limited amount of communication~\cite{barrett2007modeling}. By treating these communication-assisted classical strategies as classical circuits of limited depth, Ref.~\cite{bravyi2018quantum} showed that shallow-depth quantum circuits are more powerful than their classical counterparts. In light of these results, the past few years have seen a number of novel works laying the theoretical foundation for unconditionally demonstrating a quantum advantage  {with constant-depth quantum circuits of} nearest-neighbor entangling gates~\cite{coudron2018trading,gall2018average,watts2019exponential,grier2019interactive,bravyi2020quantum,grier2021interactive,hasegawa2021quantum}.

\begin{table*}
\centering
\caption{Summary of the nonlocal games studied in this work, along with our main experimental results using a trapped-ion quantum computer. The notation and definition of the games are provided in their respective sections. For the optimal quantum strategy, each game involves preparing the six-qubit cyclic cluster state $ \ket{C_6} $ and then evaluating a number of stabilizers from global Pauli measurement settings,  {which may differ for the games}. As the success probability $\Pr_C[\mathrm{win}]$ of a classical strategy may depend on the circuit considered, we report bounds for both depth-0 and depth-1 classical circuits. For the details of our quantum experiments, see Sec.~\ref{sec:exp_results};~the cubic Boolean function (CBF) results were obtained by a tomography experiment measuring the relevant stabilizers (data presented in Fig.~\ref{fig:stabilizer_values_bar_graph}). The stabilizer submeasurement (SS) results were obtained by a separate experiment (data presented in Table~\ref{tab:triangle_game_results}). We report experimental success probabilities $ \widehat{\Pr}_Q[\mathrm{win}] $ estimated from both the raw output of our quantum device, and after state-preparation-and-measurement (SPAM) error correction. The experimental uncertainties, denoted by the values in parentheses, correspond to a $1\sigma$ standard error within the number of significant figures reported. For each input of a game, we took $N = 5000$ shots.}
\label{tab:experimental_results}
\begin{tabular*}{\linewidth}{c @{\extracolsep{\fill}} c c c c c}
    \toprule
    {} & {Number of } & \multicolumn{2}{c}{$\Pr_C[\mathrm{win}]$} & \multicolumn{2}{c}{Experimental $\widehat{\Pr}_{Q}[\mathrm{win}]$} \\
    \cline{3-4}\cline{5-6}
    Game & (stabilizers, settings) & {Depth-0 bound} & Depth-1 bound & Raw value & SPAM-corrected \\
    \hline
    $\mathrm{CBF}(C_6,\{0,1\}^6)$ (Sec.~\ref{sec:C6_stab_game}) & $(63,63)$ & $23/32 = 71.875\%$ & $100\%$ & $80.30(8)\%$ & $83.21(9)\%$  \\
    \hline
    $\mathrm{CBF}(C_6,\mathcal{I}_{\mathrm{Mermin}}^{(55)})$ (Sec.~\ref{sec:Mermin_game}) & $(55,55)$ & $37/55 \approx 67.3\%$ & $100\%$ & $79.51(9)\%$ & $82.56(10)\%$ \\ \hline
    $\mathrm{SS}(C_6,\mathcal{I}_{\mathrm{HLF}}^{(8)})$ (Sec.~\ref{sec:SS8}) & $(5,8)$ & $7/8 = 87.5\%$ & $87.5\%$ & $81.35(59)\%$ & $85.79(53)\%$  \\ \hline
    $\mathrm{SS}(C_6,\mathcal{I}_{\mathrm{HLF}}^{(5)})$ (Sec.~\ref{sec:SS5}) & $(5,5)$ & $4/5 = 80\%$ & $80\%$ & $79.42(25)\%$ & $84.68(23)\%$  \\
    \botrule
\end{tabular*}
\end{table*}

In this paper, we analyze two kinds of Bell-type nonlocal games and demonstrate proof-of-principle experimental implementations with a six-qubit trapped-ion quantum computer~\cite{Debnath2016p63}. Both games utilize the $n$-qubit cyclic cluster state and are motivated as follows. The first game is based on the so-called graph state Bell inequality~\cite{guhne2005bell}, which we recast as the computation of a particular nonlinear Boolean function via the measurement of elements in the cyclic cluster state's stabilizer group.  
While this game has utility as estimating the state fidelity of a cyclic cluster state and benchmarking the noise of our experimental device, we also show that for an arbitrary-size instance of this game, a depth-1 classical circuit with the same gate connectivity as our quantum circuit can win the game with unit success probability. This implies that we cannot use this first game to demonstrate a quantum computational advantage in terms of circuit depth.  

This motivates the second game,
which is based on the smallest nontrivial instance of the so-called 2D hidden linear function problem introduced in Ref.~\cite{bravyi2018quantum}.  If played on a 2D grid of qubits, this game is capable of demonstrating an unconditional separation between the power of a constant-depth quantum circuits and sublogarithmic-depth classical circuits~\cite{bravyi2018quantum}.  Here we show that, on a 1D cycle, this family of games is capable of demonstrating an unconditional separation between constant-depth quantum circuits and sublinear-depth classical circuits with the same gate connectivity, which we define precisely in Theorem~\ref{thm:n_qubit_ss_games}. Furthermore, we improve upon the original game of Ref.~\cite{bravyi2018quantum} by reducing the threshold of success probability that a quantum computer must exceed down to $80\%$ from $87.5\%$.  
Our experimental implementation of the smallest instance of this second game using a 6-qubit cyclic cluster state demonstrates a success probability that is on the cusp of  {exceeding the success probability bound for} depth-1 classical circuits.  {This 6-qubit game is the first instance in a family of games that is capable of demonstrating an unconditional quantum advantage as we asymptotically scale the problem size.}
We summarize all of our results in Table~\ref{tab:experimental_results}.

The paper is organized as follows.  In Sec.~\ref{sec:background} we review background material necessary to understand our key results and establish notation.  In Sec.~\ref{sec:CBF_Games} we introduce the first type of nonlocal game, called cubic Boolean function games.  In Sec.~\ref{sec:ss_games} we discuss the second type of nonlocal game, called stabilizer submeasurement games. In Sec.~\ref{sec:exp_results} we present the results of our experimental implementation of the quantum strategies for these games.  In Sec.~\ref{sec:generalization} we discuss the generalization of these two games to the $n$-qubit scenario and formally state our claim to quantum advantage in Theorem~\ref{thm:n_qubit_ss_games}.  Finally, we conclude with a discussion and outlook in Sec.~\ref{sec:discussion}.



\section{\label{sec:background}Background}

\subsection{\label{sec:graph_states}Graph States}

Stabilizer states are a particular class of many-body quantum states that have a wide range of applications in quantum computing (e.g., for error correction, measurement-based quantum computing, and tests of Bell nonlocality). An $n$-qubit stabilizer state is defined as the joint $+1$ eigenstate of $2^n$ commuting Pauli operators, called the stabilizer group. Let $X$, $Y$, and $Z$ be the Pauli matrices. $I$ denotes the $2\times 2$ identity matrix and $\openone$ the identity operator acting on the whole system. Let $\ket{0}$ and $\ket{1}$ be the $+1$ and $-1$ eigenstates of $Z$, respectively.  Denote by $X_j$ an $n$-qubit Pauli operator that acts as $X$ on the $j$th qubit and as the identity elsewhere (similarly for $Y_j$ and $Z_j$).  Any $n$-qubit Pauli operator can be written (up to an overall phase) as $E(\mathbf{a},\mathbf{b}) = \prod_{j=0}^{n-1} i^{a_jb_j} X_j^{a_j}Z_j^{b_j}$, where $\mathbf{a},\mathbf{b}\in\{0,1\}^n$.  Denote the $j$th term in the product for $E(\mathbf{a},\mathbf{b})$ by $E_j(a_j,b_j)$.

\begin{figure*}
    \centering
    \includegraphics[width=0.8\textwidth]{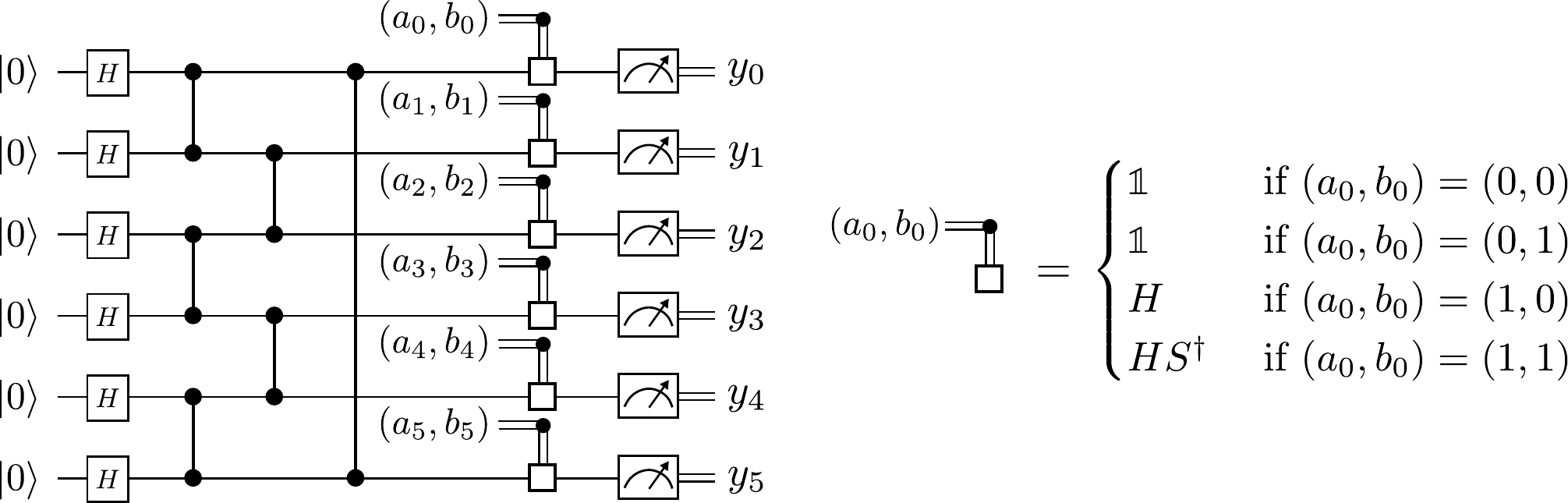}
    \includegraphics[width=0.8\textwidth]{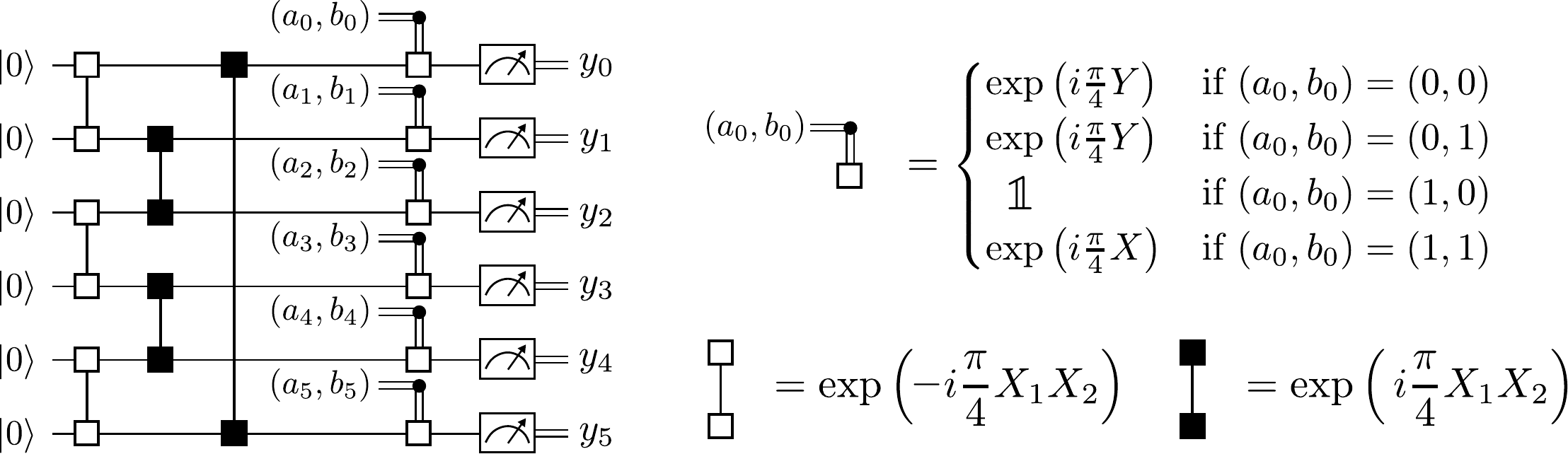}
     \caption{(Top):~The quantum circuit that prepares the six-qubit cyclic cluster state $| C_6 \rangle$ using $\CZ$ gates, followed by the measurement of the Pauli operator $E(\mathbf{a},\mathbf{b})$ for some $\mathbf{a},\mathbf{b}\in\{0,1\}^6$. Namely, the final conditional rotation changes the measurement basis of the $j$th qubit to $Z$, $X$, or $Y$ whenever the classical input $(a_j,b_j)\in\{0,1\}^2$ is $(0,1)$, $(1,0)$, or $(1,1)$, respectively.  Input $(0,0)$ denotes the presence of an identity in the Pauli operator, in which case the qubit is measured in the $Z$ basis. (Bottom):~The same circuit, recompiled using $R_{XX}(\theta)$ gates which are native to our trapped-ion device, and optimized to reduce the number of single-qubit gates.}
    \label{fig:6qubit}
\end{figure*}

It is convenient to study stabilizer states in the so-called graph-state formalism~\cite{hein2006entanglement}.  For each graph $G = (V,E)$, with vertex set $V$ and edge set $E$, the corresponding graph state $|G\rangle$ is prepared by initializing a qubit at each vertex $v\in V$ in the state $|+\rangle = (|0\rangle + |1\rangle)/\sqrt{2}$ and applying a two-qubit controlled-$Z$ gate, $\CZ_{v,w} = (\openone + Z_v + Z_w - Z_v Z_w)/2$, between each pair of vertices $(v, w) \in E$ that share an edge. That is, the graph state is defined as
\begin{equation}\label{eq:graph_state_defn}
\ket{G} = \prod_{(j,k)\in E} \CZ_{j,k} \ket{+}^{\otimes |V|}.
\end{equation}
Any stabilizer state is equivalent to a graph state up to single-qubit Clifford gates (the set of gates generated by $H=(X+Z)/\sqrt{2}$ and $S=\sqrt{Z}$)~\cite{schingelman2001graph_stabilizers,grassl2002graphs, van_den_nest2004clifford_transformations}.

Conversely, any graph state is uniquely defined in terms of its stabilizer group.  From Eq.~\eqref{eq:graph_state_defn} it follows that $|G\rangle$ is the $+1$ eigenstate of all the elements of
\begin{equation}\label{Stabilizer_Relation}
\mathcal{S}_G = \left\langle X_v \prod_{l\in\mathcal{N}(v)} Z_l ~ \bigg\vert  ~ \forall v \in V \right\rangle,
\end{equation}
where $l\in \mathcal{N}(v)$ if and only if $(l,v)\in E$.  The notation $\langle~\cdot~\rangle$ indicates the set of all possible products generated by the operators contained in the brackets. The operators in the brackets are referred to as stabilizer generators. We may associate the generator $S_v = X_v \prod_{l\in\mathcal{N}(v)} Z_l$ with the vertex $v$.

In this work we focus on the six-qubit graph state on the cycle graph $\ket{C_6}$, also known as the six-qubit cyclic cluster state, where $C_6$ denotes the six-vertex cycle graph.  Following Eq.~\eqref{eq:graph_state_defn}, the preparation of $|C_6\rangle$, followed by a measurement of an arbitrary Pauli operator $E(\mathbf{a},\mathbf{b})$, is implemented by the quantum circuit in Fig.~\ref{fig:6qubit} (top).  For the experimental implementation using trapped ions, we recompile this circuit in terms of the native two-qubit M{\o}lmer--S{\o}rensen gates, defined as $R_{XX}(\theta) = \exp\l( -i\theta X \otimes X \r)$,
yielding the circuit in Fig.~\ref{fig:6qubit} (bottom). (Further details of the experimental setup are discussed in Sec.~\ref{sec:exp_setup}.) 

\subsection{\label{sec:fid_and_gme}State fidelity and detection of entanglement}

The fidelity between a target pure state $\ket{\psi}$ and a prepared (possibly mixed) state $\rho$ is given by
\begin{equation}
 \fid(\rho,\ket{\psi}) = \bra{\psi}\rho\ket{\psi}=\tr\l{(}\op{\psi}{\psi}\rho\r{)}.
\end{equation}
Accordingly, the fidelity between $\ket{\psi}$ and $\rho$ can be thought of as the expectation value of the projector onto $\ket{\psi}$ with respect to $\rho$.  In particular, the projector onto any $n$-qubit stabilizer state, such as a graph state $\ket{G}$, can be expressed in terms of the $2^n$ elements of its stabilizer group $\S_G$,
\begin{equation}
    \op{G}{G} = \frac{1}{2^n} \sum_{S\in\S_G} S.
\end{equation}
The fidelity between a graph state $ \ket{G} $ and a state $ \rho $ is therefore determined by the expectation values of the $2^n$ elements of the stabilizer group,
\begin{equation}
    \fid(\rho, \ket{G}) = \frac{1}{2^n} \sum_{S\in\S_G} \tr(\rho S).
    \label{eq:Fidelity_Full}
\end{equation}

Another important property of the prepared state is whether or not it possesses genuine multipartite entanglement. A quantum state is said to have genuine $n$-partite entanglement if it cannot be expressed as a convex sum of biseparable states (i.e., states that are separable with respect to some bipartition of the $n$-qubit system). Since the maximum fidelity any biseparable state can have with a connected graph state is $1/2$, the operator
\begin{equation}
    \mathcal{W} = \frac{\openone}{2} - |G\rangle \langle G|
\end{equation}
acts as a witness for genuine $n$-qubit entanglement~\cite{Toth_Witnesses_Guehne_entanglement}. Namely, $\tr(\rho\mathcal{W}) < 0$ implies the $\rho$ has genuine multipartite entanglement. Thus the presence of multipartite entanglement may be detected directly from the fidelity as
\begin{equation}
\label{eq:Witness_from_Fidelity}
\tr(\rho\mathcal{W})=\frac{1}{2}-{\mathcal{F}}(\rho,\ket{G}).
\end{equation}

\subsection{\label{sec:nonlocal_bkgd}Bell-type nonlocal games}

In this section we review the concept of nonlocal games and  {translate classical strategies with communication into classical circuits.} A nonlocal game~\cite{cleve2004consequences, Brunner_review_nonlocality_2014} is a computational scenario that takes place over several rounds.  In each round, multiple parties are each provided a piece of information about a global input string $\mathbf{x}\in\textrm{Input}\subset\{0,1\}^*$ (the notation $\{0,1\}^*$ denotes the set of arbitrary length binary strings, i.e., $\{0,1\}^* = \{0,1,00,01, \ldots\}$).
The parties must then respond to a referee, each with one bit, to produce an output string $\mathbf{y}\in\{0,1\}^*$ that has certain properties depending on the input.  The referee then checks that the output possesses the desired properties and accordingly tallies the round as a win or loss.  For a particular game, denote the set of valid outputs $\mathbf{y}$ for a given input $\mathbf{x}$ as $\textrm{Win}(\mathbf{x})$.  Each such game defines a relation problem:~find a valid input--output pair $(\mathbf{x},\mathbf{y})$ satisfying some binary relation $\mathcal{R}$, whereby $(\mathbf{x},\mathbf{y})\in\mathcal{R}$ if and only if $\mathbf{y}\in\mathrm{Win}(\mathbf{x})$.

\begin{figure}
    \centering
    \includegraphics[width=0.75\linewidth]{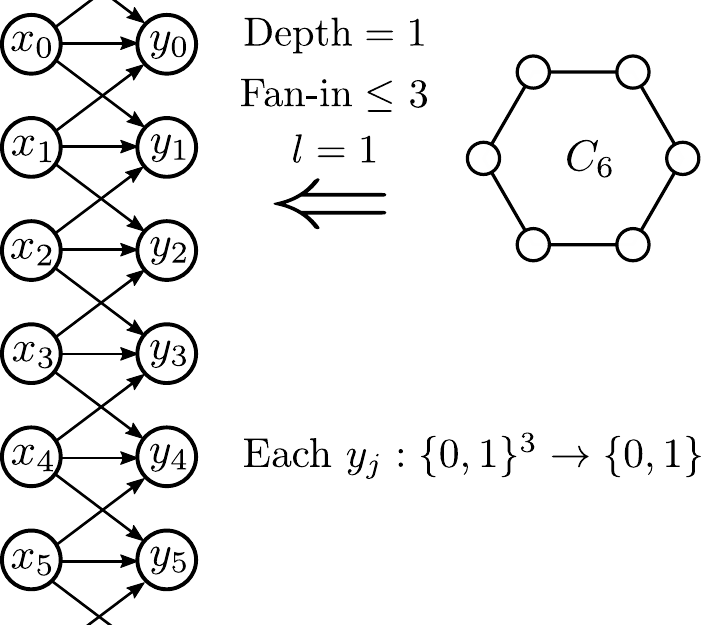}
    \caption{ {Representation} of a generic depth-1, fan-in~$\leq 3$ classical circuit that is geometrically restricted with respect to the graph $C_6$, as a directed acyclic graph. Periodic boundary conditions are imposed on the top and bottom edges on the directed acyclic graph. This particular circuit is constructed for problems where each party receives $l=1$ bit of information about the input. Each node labeled $y_j$ denotes an output of an arbitrary Boolean function of the three adjacent bits (i.e., it corresponds to an arbitrary  {three-bit} gate).   {For a general depth-$D$ circuit each output can depend on the $2D+1$ nearest inputs.}}
    \label{fig:DAG_C6_1}
\end{figure}

A strategy for a nonlocal game is a scheme by which the parties produce their outputs.  Suppose each party receives $l$ bits of information about the input $\mathbf{x}$.  In a quantum strategy, each party holds one qubit from a multipartite quantum state $\rho$, measures their qubit in some basis depending on their received bits, and then outputs the measurement outcome. We compare the quantum strategies implemented on our device (represented as quantum circuits) to classical strategies (represented as classical circuits with the same gate connectivity) in terms of their circuit depths.   {We remark that for each game studied in this work, a perfect quantum strategy  {with unit success probability} can be implemented with the constant-depth circuit depicted in Fig.~\ref{fig:6qubit}.}

In a classical strategy,  {each party may communicate} with neighboring parties to obtain a total of $k\leq K$ bits of information about the input for some integer $K$. Each party then manipulates their information via some Boolean function  $f:\{0,1\}^k\rightarrow\{0,1\}^*$ and repeats this process for a total of $D$ rounds of communication;~in the last round, each output is produced via some $y:\{0,1\}^k\rightarrow\{0,1\}$.  Each classical strategy can therefore be associated to a particular depth-$D$ classical circuit, consisting of gates drawn from an arbitrary gate set depending on no more than $K$ outputs of its neighboring gates in the previous layer (i.e., each gate is geometrically restricted with respect to some graph and has fan-in~$\leq K$), which compute the corresponding Boolean functions.  Here the depth of the classical circuit indicates the of the number of layers of  {parallelized multi-bit gates} applied. Each such circuit can be defined as a directed acyclic graph where each vertex is labeled by the corresponding gate it implements~\cite{vollmer1999introduction}.  {For example, in Fig.~\ref{fig:DAG_C6_1} the set of possible classical circuits with $D=1$ and $K=3$ that are geometrically restricted with respect to the cycle graph are visualized as a directed acyclic graph.}

To both quantum and classical strategies, one can assign an average success probability, denoted $\Pr[\mathrm{win}]$, which is the probability that the strategy will win a round of the game, given an input chosen uniformly at random.  The average success probability is computed from the conditional probability, $\Pr[\mathbf{y}|\mathbf{x}]$, for a strategy to produce output $\mathbf{y}$ given input $\mathbf{x}$, as
\begin{align}
\label{eq:success_prob_def}
    \Pr[\mathrm{win}] = \frac{1}{|\textrm{Input}|}\sum_{\mathbf{x}\in\textrm{Input}}\sum_{\mathbf{y\in\textrm{Win}(\mathbf{x})}} \Pr[\mathbf{y}|\mathbf{x}].
\end{align}
Strategies with $\Pr[\textrm{win}]=1$ are said to be perfect.  For each nonlocal game we study there exists a perfect quantum strategy. On the other hand, classical strategies for these games will produce a hierarchy of bounds depending on the depth of classical circuits. For brevity, we  express these bounds as
\begin{align}
\label{eq:ASP_Bounds_Def}
    \Pr\nolimits_C[\textrm{win}]\leq_{\textrm{Depth-0}} \beta_0 \leq_{\textrm{Depth-1}} \beta_1 \leq_{\textrm{Depth-2}} \beta_2 
    \leq
    \cdots .
\end{align}
 ``Depth-1'' and ``Depth-2'' denotes that $\beta_1$ and $\beta_2$ are the maximal values of $\Pr\nolimits_C[\mathrm{win}]$
for any geometrically restricted classical circuit with $D=1,2$, respectively, and $K=3l$. With slight abuse of notation, ``Depth-0'' denotes that $\beta_0$ is the maximum value for any classical circuit consisting of a single layer of local gates.  We say that their depth is ``zero" since these circuits correspond to classical strategies without communication, where each output $y_j$ can depend on at most $l$ bits received by the $j$th party. Surpassing any of these bounds with a quantum device is interpreted as a violation of a Bell-type inequality, which in turn demonstrates the achievement of a computational task that cannot be done with a classical circuit of that particular depth and geometry. 


\section{$C_6$ cubic Boolean function games}
\label{sec:CBF_Games}

Here we describe a nonlocal game adapted from the graph-state Bell inequalities presented in Ref.~\cite{guhne2005bell}.  In this game, parties are given partial information about an input string $\mathbf{x}\in\{0,1\}^n$ and are tasked to produce outputs such that the joint parity is equal to a particular cubic Boolean function associated to a graph $G=(V,E)$ evaluated on the input.  For a particular graph $G=(V,E)$ and input set $\mathcal{I}\subseteq\{0,1\}^n$ we denote the game as $\mathrm{CBF}(G,\mathcal{I})$.

For the sake of our six-qubit experimental demonstration, we define the game over the six-vertex cycle graph.  We remark that this six-player instance of the game contains the same qualitative features of the $n$-party scenario.  Namely, it has a non-trivial depth-0 bound, trivial depth-1 bound, and perfect quantum strategy.  For a discussion of a general size $n$ instance of the game we point the reader to Sec.~\ref{sec:generalization}.

\begin{def:C_6_Stabilizer_Game}[$\mathrm{CBF}(C_6,\mathcal{I})$]
Consider a game where six parties corresponding to the vertices of the cycle graph $C_6$ are each given a two-bit input $s_j = (a_j,b_j)\in\{0,1\}^2$, $j=0,\ldots,5$. The inputs $s_j$ are drawn as follows:~from a global input $\mathbf{x}= (x_0,\ldots,x_5)\in\mathcal{I}\subseteq\{0,1\}^6$, each party $j$ is given the bit $a_j=x_j$ and the parity of their neighbors' bits, $b_j=x_{j+1} + x_{j-1}~\mathrm{mod}~2$.  Each party $j$ then produces a one-bit output $y_j$.  The parties win the game whenever
\begin{equation}
    \sum_{j\in\mathrm{supp}(\mathbf{s})} y_j = \sum_{j=0}^5 x_{j-1}x_{j}x_{j+1}~\mathrm{mod}~2 ,
    \label{eq:c6_boolean}
\end{equation}
where $\mathrm{supp}(\mathbf{s}) = \{0\leq j\leq 5 \mid s_j\neq (0,0)\}$ and all subscripts are taken $\mathrm{mod}~6$.
\end{def:C_6_Stabilizer_Game}

In a perfect quantum strategy for this game, each party holds one qubit from the state $|C_6\rangle$ and measures the Pauli operator $E_j(a_j,b_j) = i^{a_jb_j} X_j^{a_j} Z_j^{b_j}$. 
This strategy is perfect because, collectively, the parties measure the Pauli part of the stabilizer $S_\mathbf{x} = \prod_{j=0}^5 S_j^{x_j}$, where $S_j$ denotes the stabilizer generator corresponding to vertex $j$ of $C_6$.  The parity of the measurement outcomes for qubits $j\in\mathrm{supp}(\mathbf{s})$ is then deterministically equal to the phase in front of $S_\mathbf{x}$, which is $\exp\l(i\pi\sum_{j=0}^5 x_{j-1}x_jx_{j+1}\r)$.  See Fig.~\ref{fig:6qubit} and Appendix~\ref{app:CBFn} for more details.  On the other hand, classical strategies perform with varying success depending on the particular input set $\mathcal{I}$.  We now consider two particular input sets that give two different classical bounds on the depth-0 success probability.

\subsection{\label{sec:C6_stab_game}Full-input cubic Boolean function game (fidelity game)}

Setting the input set $\mathcal{I}$ to be all possible six-bit strings, we have the game $\mathrm{CBF}(C_6,\{0,1\}^6)$. In the perfect quantum strategy presented above, each of the $2^6$ inputs corresponds to a measurement of one element of the stabilizer group for the graph state $|C_6\rangle$.  Playing many rounds of the game can be interpreted as a partial-tomography experiment to determine the fidelity of the prepared state by measuring randomly chosen stabilizer elements~\cite{flammia2011direct,dasilva2011practical}. The success probability $\Pr_Q[\mathrm{win}]$ for the quantum strategy using $\rho$ is related to the fidelity of $\rho$ with respect to $|C_6\rangle$, via
\begin{equation}
\mathcal{F}(\rho,|C_6\rangle) = 2\Pr\nolimits_Q[\mathrm{win}] - 1.
\end{equation}

For this specific input set, the depth-0 classical circuits cannot win this game with probability greater than $\Pr_C[\mathrm{win}]\leq 23/32$.  This comes from the local hidden-variable theory upper bound presented in Ref.~\cite{guhne2005bell} for the corresponding graph-state Bell inequality.  However, depth-1 circuits can implement perfect strategies. This is achieved by each party individually computing one cubic term in the function of Eq.~\eqref{eq:c6_boolean} as $y_j = a_{j-1}a_j(b_j + a_{j-1}) = x_{j-1} x_{j} x_{j+1}~\mathrm{mod}~2$, which satisfies the win conditions of the game. Hence,
\begin{align}
    \Pr\nolimits_C[\mathrm{win}]\leq_{\textrm{Depth-0}} \frac{23}{32} \leq_{\textrm{Depth-1}} 1 \leq_{\textrm{Depth-2}} 1.
\end{align}
We remark that because depth-1 classical circuits can give a perfect strategy, this Bell-type nonlocal game is not sufficient to demonstrate quantum advantage in terms of a separation in circuit depth.  Nevertheless, we use this game to estimate the fidelity of our experimentally prepared state in Sec.~\ref{sec:exp_results}.

\subsection{\label{sec:Mermin_game}Restricted-input cubic Boolean function game}

Following the seminal work of Mermin \cite{mermin1990extreme}, in Ref.~\cite{cabello2008mermin} it was shown that by restricting to a 55-element subset of $\{0,1\}^{6}$ gives the largest possible violation for any Bell inequality based on the perfect correlations present in $|C_6\rangle$. Consider the the 55-element set
\begin{equation}\label{eq:Mermin}
\mathcal{I}_{\mathrm{Mermin}}^{(55)} = \l\{\mathbf{x}\in\{0,1\}^6 \middle\vert
\begin{array}{l}
|\mathbf{x}| \neq 0, 1, \\
~\mathbf{x} \neq 010101, 101010
\end{array}
\r\},
\end{equation}
where $|\mathbf{x}| = \sum_{j=0}^5 x_j$ denotes the Hamming weight of $\mathbf{x}$ (i.e., the number of ones in the binary string). This input set yields the game $\mathrm{CBF}(C_6,\mathcal{I}_{\mathrm{Mermin}}^{(55)})$.

The same quantum strategy presented for $\mathrm{CBF}(C_6,\{0,1\}^6)$ is also perfect for this game. For the classical bounds, depth-0 classical circuits cannot win with probability greater than $37/55$, although the same depth-1 circuit presented in Sec.~\ref{sec:C6_stab_game} remains perfect. Hence,
\begin{align}
    \Pr\nolimits_C[\mathrm{win}]\leq_{\textrm{Depth-0}} \frac{37}{55} \leq_{\textrm{Depth-1}} 1 \leq_{\textrm{Depth-2}} 1.
\end{align}
Again, the trivial depth-1 bound indicates that this Bell-type nonlocal game is not sufficient to demonstrate quantum advantage in terms of a separation in circuit depth.

\section{$C_6$ stabilizer submeasurement games}
\label{sec:ss_games}

We now move on to another class of games called stabilizer submeasurement games.  For the sake of our experimental demonstration, we first discuss this game for a fixed-sized input of six bits.  Unlike the CBF games, no depth-1 classical strategy is perfect for this six-bit instance.  This property begins to reveal an important feature of this game for general $ n $-bit instances. Namely, no strategy produced by a geometrically-restricted classical circuits with depth growing sublinearly in $n$ is perfect. We make this statement precise in Theorem~\ref{thm:n_qubit_ss_games}.

Let us first see how this property manifests for the six-bit input.
These games are defined as follows.
\begin{def:Stab_Sub_Game}[$\mathrm{SS}(C_6,\mathcal{I})$]
Consider a game where six parties corresponding to the vertices of the cycle graph $C_6$ are each given a one-bit input $x_j$, $j=0,\ldots,5$, from a global input $\mathbf{x}=(x_0,\ldots,x_5)\in\mathcal{I}\subseteq\{0,1\}^6$, drawn at random from a uniform distribution. 
With each input we can associate a global Pauli operator $E(\mathbf{1},\mathbf{x}) = \prod_{j=0}^{5} i^{x_j} X_j Z_j^{x_j}$, where $\mathbf{1}$ denotes the all-ones string. Each party then produces an output $y_j\in\{0,1\}$, which collectively forms a string $\mathbf{y}\in\{0,1\}^6$.  The parties are said to win the game whenever
\begin{subequations}\label{eq:SS_win_conditions_6}
\begin{align}
&\forall P\subseteq E(\mathbf{1},\mathbf{x}) ~\mathrm{s.t.}~P\in\mathcal{S}_{C_6}~\sum_{j\in\mathrm{supp}(P)} y_j = 0, \\
&\forall P\subseteq E(\mathbf{1},\mathbf{x}) ~\mathrm{s.t.}~P\in-\mathcal{S}_{C_6}~\sum_{j\in\mathrm{supp}(P)} y_j = 1 ,
\end{align}
\end{subequations}
where $P\subseteq E(\mathbf{1},\mathbf{x})$ means that $P$ is a Pauli operator obtained by replacing some nonidentity tensor factors in $E(\mathbf{1},\mathbf{x})$ with the identity.  Furthermore, $\mathrm{supp}(P) = \mathrm{supp}(\mathbf{a},\mathbf{b})$ where $P = E(\mathbf{a},\mathbf{b})$ (i.e., $\mathrm{supp}(P)$ indexes the qubits where $P$ acts nontrivially) and all arithmetic is performed modulo 2.
\end{def:Stab_Sub_Game}
The $\mathrm{SS}(C_6,\mathcal{I})$ games always have a perfect quantum strategy wherein each party $j$ holds one qubit from the state $|C_6\rangle$, measures it in the basis of $X_j$ or $Y_j$ if $x_j=0 \text{ or } 1$, respectively, and then outputs the measurement outcome. For more details on the measurement outcomes expected when a stabilizer element exists as a submeasurement of a global Pauli measurement, see Appendix~\ref{sec:stab_sub_details}.  On the other hand, classical strategies perform with varying success depending on the particular input set $\mathcal{I}$.

\subsection{\label{sec:SS8}2D hidden linear function (HLF) game} 

Fixing the input set to be
\begin{equation}\label{eq:HLF_8}
\mathcal{I}_{\mathrm{HLF}}^{(8)} = \l\{\mathbf{x}\in\{0,1\}^6 \middle\vert
\begin{array}{l}
(x_0, x_{2}, x_{4}) \in \{0,1\}^3, \\
x_j = 0 \text{ for all other } j
\end{array}
\r\},
\end{equation}
we have the game $\mathrm{SS}(C_6,\mathcal{I}_{\mathrm{HLF}}^{(8)})$.  This game is equivalent to the smallest instance of the 2D hidden linear function game introduced in Ref.~\cite{bravyi2018quantum} that was used to demonstrate an unconditional exponential separation between the power of classical and quantum circuits.

As shown in Appendix~\ref{sec:Triangle_Game_Bounds}, all geometrically restricted depth-0 and depth-1 fan-in~$\leq 3$ circuits cannot win this game on more than $7/8$ inputs; however, there is a depth-2 circuit that wins on all inputs. Hence,
\begin{align}
    \Pr\nolimits_C[\mathrm{win}]\leq_{\textrm{Depth-0}} \frac{7}{8} \leq_{\textrm{Depth-1}} \frac{7}{8} \leq_{\textrm{Depth-2}} 1.
\end{align}
Thus, for the current 6-qubit example, the {constant-depth} quantum circuit in Fig.~\ref{fig:6qubit} can in principle achieve the higher success probability in comparison to any geometrically restricted classical circuit  {of depth 1}. This nontrivial classical depth-1 bound is the first of many nontrivial bounds we can put on classical circuits with increasing depth as we scale the size of the input. This fact makes this game lucrative for demonstrating quantum advantage in terms of a separation in circuit depth as we modestly scale the system size.   {We will further explore this point in Sec.~\ref{sec:generalization}.}

\subsection{\label{sec:SS5}Restricted-input 2D hidden linear function (HLF) game }

Analogous to the case for the CBF game, we ask whether the quantum violation of the classical success probability bound can be increased by further restricting the input set. 
Defining the input set
\begin{equation}\label{eq:HLF_5}
\mathcal{I}_{\mathrm{HLF}}^{(5)} = \l\{\mathbf{x}\in\{0,1\}^6 \middle\vert
\begin{array}{l}
(x_0, x_{2}, x_{4}) \in \mathcal{V}, \\
x_j = 0 \text{ for all other } j
\end{array}
\r\},
\end{equation}
where $\mathcal{V} \subset \{0,1\}^3$ is defined as
\begin{equation}
    \mathcal{V} = \{(0,0,0), (0,1,1), (1,0,1), (1,1,0), (1,1,1)\},
\end{equation}
we obtain the game $\mathrm{SS}(C_6,\mathcal{I}_{\mathrm{HLF}}^{(5)})$.  Analogous to how a restricted input set makes the game $\mathrm{CBF}(C_6,\mathcal{I}_\mathrm{Mermin}^{(55)})$ classically harder than $\mathrm{CBF}(C_6,\{0,1\}^6)$ (i.e., the depth-0 lower bound on $\Pr_C[\mathrm{win}]$ decreases), the same behavior occurs when restricting from $\mathcal{I}_{\mathrm{HLF}}^{(8)}$ to $\mathcal{I}_{\mathrm{HLF}}^{(5)}$.  As shown in Appendix~\ref{sec:Triangle_Game_Bounds}, the depth-0 and depth-1 classical bounds are reduced to
\begin{align}
    \Pr\nolimits_C[\mathrm{win}]\leq_{\textrm{Depth-0}} \frac{4}{5} \leq_{\textrm{Depth-1}} \frac{4}{5} \leq_{\textrm{Depth-2}} 1.
\end{align}
We remark that this game improves on the previously known bound on the success probability~\cite{bravyi2018quantum}, leaving more room for noise in the quantum strategy.



\section{\label{sec:exp_results}Experimental results}

\subsection{\label{sec:exp_setup}Experimental setup}

The experiments presented here were performed on a fully programmable trapped-ion quantum computer~\cite{Debnath2016p63}. The apparatus is based on a linear chain of $^{171}$Yb$^{+}$ ions confined in a Paul trap, with each qubit encoded in two hyperfine states of the $^{2}S_{1/2}$ ground-state manifold. Prior to implementing a quantum circuit, the ions are ground-state cooled and initialized to the $\vert 0 \rangle$ state with optical pumping~\cite{Olmschenk2007p052314}.

Our device has a universal gate set consisting of two classes of quantum operations:~single-qubit rotations and two-qubit entangling interactions ($R_{XX}$ gates). These operations are achieved by applying two counter-propagating optical Raman beams derived from a pulsed 355-nm mode-locked laser~\cite{Islam2014p3238}. One Raman beam is a global beam applied to the entire chain, while the other is split into an array of individual addressing beams, each of which can be controlled independently and targets a single qubit. Single-qubit rotations around the $z$-axis are achieved by phase advances on the classical control signals, while single-qubit rotations around axes in the $xy$-plane are realized by driving resonant Rabi rotations of defined phase, amplitude, and duration. Two-qubit gates are implemented by illuminating two selected ions with frequencies near the motional sidebands, creating an effective Ising spin--spin interaction via transient entanglement between the two qubits and the motion in the trap~\cite{Molmer1999p1835,Solano1999pR2539}. We use multi-frequency pulses to ensure the qubit states are disentangled from the motional modes at the end of the gate~\cite{Blumel2021efficient}. The angle of rotation for both the single-qubit gates and $R_{XX}$ gates can be varied continuously.

\begin{figure*}
    \centering
    \includegraphics[width=0.95\textwidth]{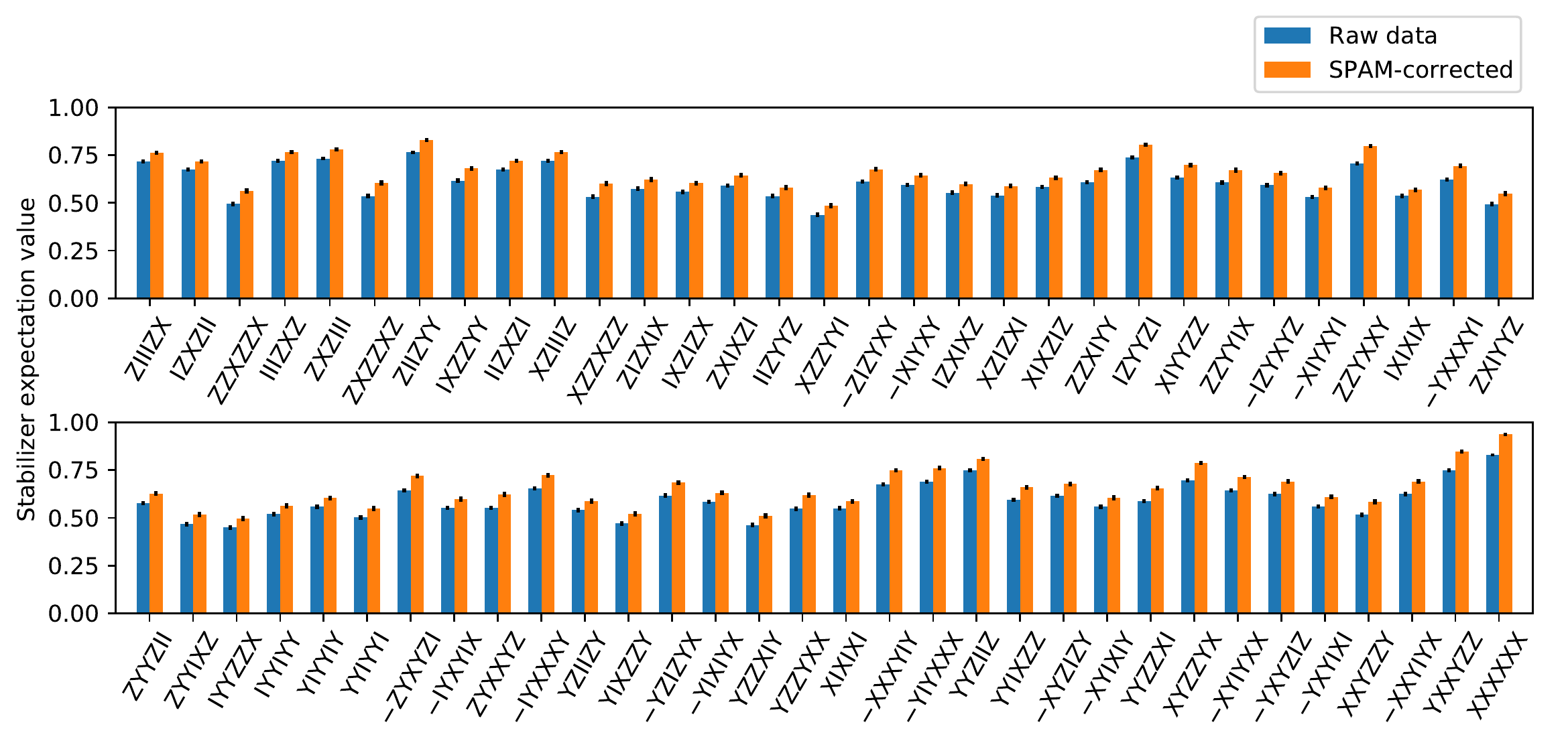}
    \caption{Estimated $|C_6\rangle$ stabilizer expectation values from our ion-trap experiment. This data allows us to evaluate the fidelity of our prepared state [Eq.~\eqref{eq:C6_fidelity}] and win probabilities for the cubic Boolean function games $\mathrm{CBF}(C_6,\{0,1\}^6)$ and $\mathrm{CBF}(C_6,\mathcal{I}_{\mathrm{Mermin}}^{(55)})$ [Eqs.~\eqref{eq:C6_win_prob} and \eqref{eq:mermin_win_prob}, respectively]. The identity stabilizer is not included, as it trivially satisfies $ \tr(\rho\openone) = 1 $. Each expectation value was estimated from $ N = 5000 $ shots, from which error bars of $ 1\sigma $ statistical uncertainties are reported. For explicit numerical values, as well as the grouping of simultaneously measured stabilizers that we employed, see Table~\ref{tab:stabilizer_exp_values} in Appendix~\ref{sec:additional_data}.}
    \label{fig:stabilizer_values_bar_graph}
\end{figure*}

Typical single- and two-qubit gate fidelities are $99.0(5)\%$ and $98.5(5)\%$, respectively. The latter is limited by residual entanglement of the qubit states and the motional state of the ions due to intensity noise and motional heating. Immediately before running one of the experiments, a lower bound $A$ on the infidelity of every two-qubit gate is estimated to ensure successful calibration. For each pair of qubits, this quantity is determined by applying a (noisy) $R_{XX}$ gate to the computational basis state $|00\rangle$ and measuring the odd-parity-state population $A=\bra{01}\rho\ket{01}+\bra{10}\rho\ket{10}$ of the resulting two-qubit state $\rho$~\cite{figgatt2019parallel}. 

The implementation of a circuit consists of sequences of gate operations between qubits, compiled into the native single- and two-qubit gate operations. After executing the quantum circuit, the state of each ion is read out in parallel using state-dependent fluorescence detection~\cite{Olmschenk2007p052314}. State-preparation-and-measurement (SPAM) errors are characterized and corrected by applying the inverse of an independently measured state-to-state error matrix~\cite{Shen2012p053053}. Detailed performance of the system has been described elsewhere~\cite{Debnath2016p63,Landsman2019p61}.

\subsection{Full-input cubic Boolean function game (fidelity game)}

As discussed in Sec.~\ref{sec:C6_stab_game}, the fidelity $ \fid(\rho, \ket{C_6}) $ indicates the probability of winning the game $\mathrm{CBF}(C_6,\{0,1\}^6) $, using $ \rho $ as the quantum resource. We therefore directly estimate this fidelity by performing tomography over all $ 2^6 - 1 = 63 $ nontrivial stabilizers. A straightforward implementation is to measure the 63 stabilizers one-by-one, with an equal number of samples each. Indeed, this would be essentially equivalent to known probabilistic protocols for estimating state fidelities~\cite{flammia2011direct,dasilva2011practical}, which in this context use uniformly randomized measurement settings over all stabilizers. In order to obtain an estimate of the fidelity within precision $\pm\epsilon$, with success probability $\geq 1-\delta$, such methods require up to $ \lceil 8 \ln(4/\delta) / \epsilon^2 \rceil $ measurements. This is notably independent of system size, despite the fact that the total number of stabilizers grows exponentially with $ n $. However, since our six-qubit system is fairly small, this randomization turns out to incur a larger constant overhead compared to a more direct approach.

Therefore we instead use a greedy graph-coloring heuristic to find locally commuting subsets of the stabilizers~\cite{verteletskyi2020measurement,jena2019pauli}. Two Pauli operators $ \bigotimes_{j=0}^{n-1} P_j, \bigotimes_{j=0}^{n-1} Q_j \in \{I, X, Y, Z\}^{\otimes n} $ are said to locally commute if $ [P_j, Q_j] = 0 $ for all $ j $. This allows us to parallelize the stabilizer measurements into fewer measurement settings, reducing the overall sampling runtime. For this experiment, we used a total of 37 global Pauli measurement settings to estimate all 63 nontrivial stabilizer expectation values (see Table~\ref{tab:stabilizer_exp_values} in Appendix~\ref{sec:additional_data} for details of this partitioning). Each measurement setting was sampled $ 5000 $ times to obtain adequate statistics.

The stabilizer expectation values are visualized in Fig.~\ref{fig:stabilizer_values_bar_graph}, with $ 1\sigma $ uncertainty bars associated with the sampling statistics. Explicit numerical values are provided in Table~\ref{tab:stabilizer_exp_values} in the Appendix. We include the values from both the raw experimental data, as well as after applying SPAM correction to the outcome distributions. From these results we estimate the fidelity of our prepared state $ \rho $ to be
\begin{equation}\label{eq:C6_fidelity}
	\widehat{\fid}(\rho, \ket{C_6}) =
	\begin{cases}
		0.6061(17) & \text{from raw data}\\
		0.6639(19) & \text{SPAM-corrected}.
	\end{cases}
\end{equation}
The statistical uncertainty is computed by propagating uncertainties from the stabilizer expectation values, paying special care to handle the covariances between stabilizers that are simultaneously evaluated within the same Pauli measurement setting. Details of this error analysis are provided in Appendix~\ref{sec:uncert_analysis}.

To put this result in context, we briefly review prior experimental work on graph states over six qubits. In Ref.~\cite{lu2007experimental}, an H-graph state was prepared with fidelity $0.593(25)$ on a photonic device. More recently, Ref.~\cite{bravyi2020mitigating} prepared a six-qubit path-graph state with fidelity just under $0.6$ on a superconducting architecture;~using the error-mitigation techniques also developed in that work, they then demonstrated an error-mitigated fidelity above $0.8$. In both works, the underlying graphs required only five two-qubit entangling gates, whereas our $\ket{C_6}$ state requires six such gates.

As discussed in Sec.~\ref{sec:fid_and_gme}, we can also use the fidelity of $ \rho $ to compute an estimate for the genuine-multipartite-entanglement witness $ \mathcal{W} $ via Eq.~\eqref{eq:Witness_from_Fidelity}. From the data, we obtain the estimate
\begin{align}
    \widehat{\tr(\rho\mathcal{W})} = \begin{cases}
		-0.1061(17) & \text{from raw data}\\
		-0.1639(19) & \text{SPAM-corrected}
	\end{cases}.
\end{align}
The value of the witness is negative, which indicates the presence of genuine multipartite entanglement in our experimentally prepared state. In other words, we guarantee that $\rho$ is entangled across all possible bipartitions of our six-qubit system.

Finally, using the relation $ \fid(\rho, \ket{C_6}) = 2 \Pr_Q[\text{win}] - 1 $ for the CBF games, we obtain a raw success probability for winning this full-input game with our experimentally prepared state $\rho$ as
\begin{equation}\label{eq:C6_win_prob}
    \widehat{\Pr}_Q[\text{win}] = 80.30(8) \%. \quad [\mathrm{CBF}(C_6,\{0,1\}^6) \text{ game}]
\end{equation}
This value significantly surpasses the classical bound of $ 23/32 \approx 72\% $.

Note that for our experimental win probabilities, we report only the values obtained from the raw data. We argue that any conclusions drawn from SPAM-corrected data in this context may be controversial. This is because nonlocal games are evaluated on a shot-to-shot basis:~a given input--output pair either does or does not satisfy the requisite win conditions. Equivalently, from a computational perspective we require that the quantum computer evaluate the correct value of the cubic Boolean function given a single input string. The final success rate that we quote is merely an average over many individual instances. On the other hand, SPAM-error correction is a postprocessing technique that is necessarily applied to the global distribution observed over many shots. This technique allows us to separate the fidelity of our cluster-state preparation from extraneous device errors such as detector readout noise. For this reason we have included SPAM-corrected estimates for the fidelity and entanglement witness, as they illustrate the quality of the entangling circuit by itself before further corruption by measurement noise. However, when playing a nonlocal game, players are ultimately restricted to claiming either a win or loss at each shot, so measurement errors are inevitable in this context.

\subsection{Restricted-input cubic Boolean function game}

Using the same experimental data, we can also estimate our device's average success probability for the game $\textrm{CBF}(C_6,\mathcal{I}_{\mathrm{Mermin}}^{(55)})$ discussed in Sec.~\ref{sec:Mermin_game}. This game was also recently implemented in Ref.~\cite{sheffer2021playing} using IonQ's cloud-accessible quantum computer, where a winning probability of $87(1)\%$ was reported.
Using the data, we estimate
\begin{equation}\label{eq:mermin_win_prob}
	\widehat{\Pr}_Q[\mathrm{win}] = 79.51(9)\%, \  [\mathrm{CBF}(C_6,\mathcal{I}_{\mathrm{Mermin}}^{(55)}) \text{ game}]
\end{equation}
far surpassing the classical bound of $37/55 \approx 67\%$. We remind the reader that this result corresponds to the raw output of our quantum device.

\subsection{2D hidden linear function game}

\begin{table}
\caption{Experimental results for the stabilizer submeasurement games, $\mathrm{SS}(C_6,\mathcal{I}_{\mathrm{HLF}}^{(8)})$ and $\mathrm{SS}(C_6,\mathcal{I}_{\mathrm{HLF}}^{(5)})$. The rules and winning conditions of the game are described in Sec.~\ref{sec:ss_games}. For each input, $ 5000 $ rounds of the game were played, from which we estimate a win rate and a $ 1\sigma $ statistical uncertainty. For completeness we include the SPAM-corrected probabilities, although as argued in the main text, the proper figure of merit for claiming quantum computational advantage here (or lack thereof) is the raw value.}
\label{tab:triangle_game_results}
\begin{tabular*}{\linewidth}{c @{\extracolsep{\fill}} c c c}
    \toprule
     & \multicolumn{2}{c}{Win rate (\%)} &  Classical \\
    \cline{2-3}
    Input & Raw value & SPAM-corrected & bound (\%) \\
    \hline
    $ 000000 $ & $ 80.42(56) $ & $ 86.45(48) $ & - \\
    $ 000010 $ & $ 84.76(51) $ & $ 87.85(46) $ & - \\
    $ 001000 $ & $ 83.92(52) $ & $ 86.94(48) $ & - \\
    $ 001010 $ & $ 75.58(61) $ & $ 81.12(55) $ & - \\
    $ 100000 $ & $ 85.04(50) $ & $ 88.11(46) $ & - \\
    $ 100010 $ & $ 78.12(58) $ & $ 83.98(52) $ & - \\
    $ 101000 $ & $ 78.46(58) $ & $ 84.28(51) $ & - \\
    $ 101010 $ & $ 84.52(51) $ & $ 87.56(47) $ & - \\
    \hline
    $\mathcal{I}_{\mathrm{HLF}}^{(8)}$ & $ 81.35(59) $ & $ 85.79(53) $ & $ 87.5 $ \\
    $\mathcal{I}_{\mathrm{HLF}}^{(5)}$ & $ 79.42(25) $ & $ 84.68(23) $ & $ 80 $ \\
    \botrule
\end{tabular*}
\end{table}

As another benchmark of our device, we performed the game $\mathrm{SS}(C_6,\mathcal{I}_{\mathrm{HLF}}^{(8)})$, described in Sec.~\ref{sec:SS8}, which is the smallest example of the so-called 2D hidden linear function game of Ref.~\cite{bravyi2018quantum}. This game was also recently implemented in Ref.~\cite{sheffer2021playing} on IBM, IonQ, and Honeywell's cloud-accessible quantum computers. Their highest reported success probability was $ 85(1)\% $, using Honeywell's H0 trapped-ion device, which does not surpass the depth-1 classical bound of $ 87.5\% $.

In this experiment, we played $ 5000 $ rounds of the game per input to determine the average success probability for our device.  The results for each input are presented in Table~\ref{tab:triangle_game_results}. The estimated average success probability is
\begin{align}
	\widehat{\Pr}_Q[\mathrm{win}] &= 81.35(59)\%, \quad [\mathrm{SS}(C_6,\mathcal{I}_{\mathrm{HLF}}^{(8)}) \text{ game}] \label{eq:8-input_ss}
\end{align}
which is below the depth-1 classical bound.

\subsection{Restricted-input 2D hidden linear function game}

Using the same experimental data from Table~\ref{tab:triangle_game_results}, we also estimate our device's average success probability for the game $\mathrm{SS}(C_6,\mathcal{I}_{\mathrm{HLF}}^{(5)})$, described in Sec.~\ref{sec:SS5}. The estimated average success probability in this case is
\begin{align}
    \widehat{\Pr}_Q[\mathrm{win}] &= 79.42(25)\%. \quad [\mathrm{SS}(C_6,\mathcal{I}_{\mathrm{HLF}}^{(5)}) \text{ game}], \label{eq:5-input_ss}
\end{align}
which is very close to meeting the depth-1 classical bound of $80\%$.

\subsection{\label{sec:device_char}Characterization of device noise}

\begin{figure*}
    \centering
    \includegraphics[width=\linewidth]{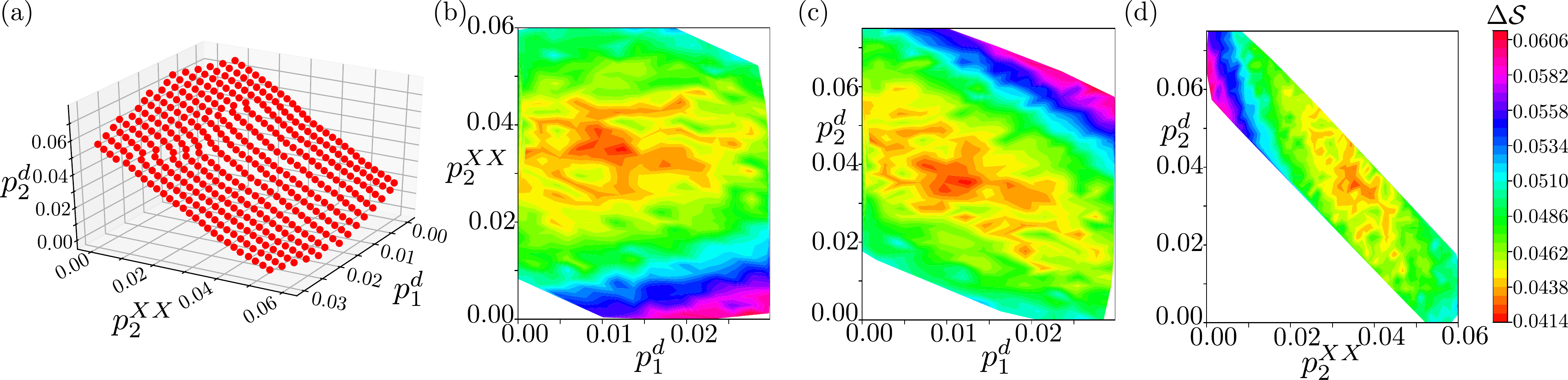}
    \caption{Characterization of noise parameters of the device, using the SPAM-corrected experimental data. The variational parameters of our noise model, described in Sec.~\ref{sec:device_char}, control the strength of single-qubit depolarizing channels (noise rate $p_1^d$), two-qubit depolarizing channels ($p_2^d$), and two-qubit joint bit-flip channels ($p_2^{XX}$). 
    (a)~The surface of $\Delta \fid = 0$, as defined in Eq.~\eqref{eq:Delta_Fid}. (b)--(d)~Contour plots of $\Delta\mathcal{S}$ over that surface, interpolated over a finer mesh. To visualize this surface, we project onto either the (b)~$(p_1^d,p_2^{XX})$ plane, (c)~$(p_1^d,p_2^d)$ plane, or (d)~$(p_2^{XX},p_2^d)$ plane. The minimum value for $\Delta \mathcal{S}$ is obtained at $(p_1^d,p_2^{XX},p_2^d) = (1.2\%,3.5\%,3.5\%)$, corresponding to the best-fit parameters for this noise model of our device.}
    \label{fig:Countour_simulation_values}
\end{figure*}

In this section we present a method for characterizing the various types of noise in our trapped-ion device based on the SPAM-corrected stabilizer information (cf.~Fig.~\ref{fig:stabilizer_values_bar_graph}) and classical simulations. Since the experiments we implemented involved only six qubits, a single instance of the experiment with a particular set of parameters can be simulated on a laptop in a few minutes, 
making the following method tractable for small system sizes. We simulate noisy quantum circuits with a physically motivated error model over many different noise rates and determine the values which best fit the experimental data.

\subsubsection{Error model}

The error model we chose consists of three noise channels with variable error rates, and two channels with fixed error rates that were determined from device calibration beforehand.  The variable noise channels consist of a single-qubit depolarizing error after every single-qubit gate, and both a two-qubit depolarizing error and a stochastic joint bit-flip error after every two-qubit $R_{XX}$ gate. The corresponding error channels (acting on qubits $j, k$) are defined as
\begin{equation}
    \mathcal{E}_{1d}^{(j)}(\rho) = (1-p_1^d) \rho + \frac{p_1^d}{3}(X_j\rho X_j + Y_j\rho Y_j + Z_j \rho Z_j),\label{eq:1d_error}\\
\end{equation}
\begin{equation}
    \mathcal{E}_{2d}^{(j,k)}(\rho) = (1-p_2^d) \rho + \frac{p_2^d}{15}\sum_{ (P,Q)\in\mathcal{K}_2^\mathrm{dep}   } P_j Q_k \rho Q_k P_j,\label{eq:2d_error}
\end{equation}
\begin{equation}
    \mathcal{E}_{XX}^{(j,k)}(\rho) = (1-p_2^{XX})\rho + p_2^{XX} X_j X_k \rho X_k X_j,\label{eq:XX_error}
\end{equation}
where $\mathcal{K}_2^\mathrm{dep} = \{I,X,Y,Z\}^2\setminus\{(I,I)\}$  and $p_1^d$, $p_2^d$, and $p_2^{XX}$ are the characteristic noise parameters, respectively. The use of two types of two-qubit gate errors is motivated by the physics of the ion-trap device. The $XX$ error channel can be interpreted as the ensemble average of random, normally distributed overrotations of the $R_{XX}$ gate (with variance related to $p_2^{XX}$)~\cite{debroy2020logical}. However this type of error preserves two-qubit parity, so we additionally include the depolarizing error channel as an effective model of more general noise features.

For the fixed error channels, we assume that a single-qubit dephasing error occurs whenever a qubit idles in the circuit, and a two-qubit crosstalk error occurs between ion pairs that are unintentionally coupled when we apply an $R_{XX}$ gate~\cite{debroy2020logical}.  Idling errors are modeled by the channel
\begin{equation}
    \mathcal{E}_{\mathrm{idle}}^{(j)}(\rho) = (1-p_{\mathrm{idle}}) \rho + p_{\mathrm{idle}} Z_j \rho Z_j,
\end{equation}
where the error rate depends on the gate time $t$ as $p_{\mathrm{idle}} = (1-e^{t/2T_2})/2$. Based on prior calibration of our device, we estimate $T_2 = \SI{200}{\milli\second}$. The single- and two-qubit gate times are $t_1 = \SI{10}{\micro\second}$ and $t_2 = \SI{350}{\micro\second}$, respectively. Crosstalk errors are modeled with $XX$ errors [see Eq.~\eqref{eq:XX_error}] between each pair of unintentionally coupled ions with error rate $p_c = 0.06\%$, as determined by the crosstalk Rabi ratio $\Omega_c / \Omega_R = 3\%$ of our device.  Measurement noise is set to zero, as we compare the outcome of the classical simulation to the SPAM-corrected data, which ideally corrects for all measurement errors.  We use Qiskit~\cite{Qiskit-Textbook} to simulate the noisy circuit using the quantum-trajectories method, i.e., by stochastically applying Kraus operators with their corresponding probabilities at the appropriate steps in the circuit. We take $10^4$ shots per measurement setting for each simulation point, so that sampling errors are negligible.

\subsubsection{Identifying noise parameters of best fit}

Given the noise model, we classically simulate the noisy circuit for a range of noise parameters $(p_1^d,p_2^{XX},p_2^d)$ and find the values that most closely approximate the experiment by analyzing two figures of merit. The first is the absolute difference between the fidelity value obtained from the experiment, $\widehat{{\fid}}_\mathrm{exp}$, and the fidelity value obtained from classical simulation, $\fid_\mathrm{sim}$, 
\begin{equation}
\label{eq:Delta_Fid}
\Delta { \fid} = |\widehat{{\fid}}_\mathrm{exp}-\fid_\mathrm{sim}|=\frac{1}{64}\l\lvert \sum_{S\in\mathcal{S}_{C_6}} \ev{S}_\mathrm{exp}-\ev{S}_\mathrm{sim} \r\rvert.
\end{equation}
Here, $\mathcal{S}_{C_6}$ is the stabilizer group of $\ket{C_6}$, and $\ev{S}_\mathrm{exp}$ and $\ev{S}_\mathrm{sim}$ denote the expectation value of stabilizer $S$ calculated from the experimental  and simulated data, respectively. The second metric is the average absolute difference of stabilizer values,
\begin{equation}
\label{eq:Delta_Stab}
\Delta \mathcal{S} = \frac{1}{64}\sum_{S\in\mathcal{S}_{C_6}} |\ev{S}_\text{exp}-\ev{S}_\text{sim}|.
\end{equation}
If our simulation exactly mirrors the experiment, then $\Delta \fid = \Delta \mathcal{S}=0 $.

For our task of fitting the noise parameters, we employ both metrics. The weaker $\Delta\fid=0$ condition works at the level of coarse-graining. We simulate the experiment over a grid of physically reasonable $(p_1^d,p_2^{XX},p_2^d)$ values and interpolate to create a 3D density plot for $\Delta \fid(p_1^d,p_2^{XX},p_2^d)$. From this, we find the surface corresponding to $\Delta \fid = 0$, shown in Fig.~\ref{fig:Countour_simulation_values}(a). We then employ $\Delta \mathcal{S}$ for a more fine-grained analysis. Simulating the experiment over a mesh of $(p_1^d,p_2^{XX},p_2^d)$ values on the $\Delta\fid = 0$ surface, we can identify the region where $\Delta \mathcal{S}$ is minimized.
To better visualize the behaviour of $\Delta \mathcal{S}$, we interpolate the simulated data points and project the surface onto one of the three planes formed by the $p_1^d,p_2^{XX},p_2^d$ axes, as shown in Fig.~\ref{fig:Countour_simulation_values}(b)--(d). From these figures, we estimate the best-fit parameters to be $(p_1^d,p_2^{XX},p_2^d) = (1.2\%,3.5\%,3.5\%)$. It can be shown that the gate infidelities~\cite[Sec.~9.3]{nielsen_chuang_2010} of our three error channels [Eqs.~(\ref{eq:1d_error}--\ref{eq:XX_error})] are $p_1^d$, $p_2^{XX}$, and $\frac{4}{5}p_2^{d}$ respectively, corresponding to an overall two-qubit gate infidelity of $p_2^{XX}+\frac{4}{5}p_2^{d}=6.3\%$.

To put this result in context, consider the calibration procedure which produces a lower bound $A$ on the two-qubit gate infidelity (described in Sec.~\ref{sec:exp_setup}). For this experiment, we had measured an average value of $A$ over the six $R_{XX}$ gates to be 1.8\%. This coincides well with our characterization, as a depolarizing channel with strength $p_2^d=3.5\%$ would result in $A = 1.9\%$. Note that the $XX$ errors of our model cannot affect the value of $A$, as they cannot change the parity of computational basis states. 
Hence, the two-qubit depolarizing noise can account for the calibrated value of $A$ and the infidelity of a standalone $R_{XX}$ gate. 
However, our results from implementing the full circuit of Fig.~\ref{fig:6qubit} indicate a degradation of gate fidelities relative to the calibrated values. The $XX$ channel of our error model is included to account for this additional effect, which one may interpret as random overrotations due to laser-intensity fluctuations. 

We also repeated this analysis using purely depolarizing noise (i.e., fixing $p_2^{XX}=0$), obtaining a best-fit value of $\Delta \mathcal{S} = 0.05$. This is larger than the value of $\Delta\mathcal{S} = 0.04$ obtained in Fig.~\ref{fig:Countour_simulation_values}, indicating that, unsurprisingly, our device-specific noise model works better than a device-agnostic depolarizing model to reproduce the experimental data.



\section{\label{sec:generalization}$n$-qubit generalization}

While our experimental results were limited to a six-qubit system, here we discuss quantum and classical strategies for the two types of games with an arbitrary sized input.   {In this case,} the the cubic Boolean function games always have a perfect depth-1 classical strategy, whereas geometrically-restricted classical circuits of depth scaling sublinearly with the input size fail to win the stabilizer submeasurement games with probability greater than 80\%.  We make the latter statement precise in Theorem~\ref{thm:n_qubit_ss_games}.


\subsection{$C_n$ cubic Boolean function games}
The premise and win conditions $C_6$ cubic Boolean function games introduced in Sec.~\ref{sec:CBF_Games} generalize to the $n$-qubit scenario as follows.

\begin{def:C_n_Stabilizer_Game}[$\mathrm{CBF}(C_n,\mathcal{I})$]
Consider a game where $n$ parties corresponding to the vertices of the cycle graph $C_n$ are each given a two-bit input $s_j = (a_j,b_j)\in\{0,1\}^2$, $j=0,\ldots,n-1$. The inputs $s_j$ are drawn as follows:~from a global input $\mathbf{x}= (x_0,\ldots,x_{n-1})\in\mathcal{I}\subseteq\{0,1\}^n$, each party $j$ is given the bit $a_j=x_j$ and the parity of their neighbors' bits, $b_j=x_{j+1} + x_{j-1}~\mathrm{mod}~2$.  Each party $j$ then produces a one-bit output $y_j$.  The parties win the game whenever
\begin{align}
    \sum_{j\in\mathrm{supp}(\mathbf{s})} y_j = \sum_{j=0}^{n-1}x_{j-1}x_{j}x_{j+1} ~\mathrm{mod}~2 ,
    \label{eq:c6_boolean_n}
\end{align}
where $\mathrm{supp}(\mathbf{s}) = \{0\leq j\leq {n-1} \mid s_j\neq (0,0)\}$ and all subscripts are taken $\mathrm{mod}~n$.
\end{def:C_n_Stabilizer_Game}

As shown in Appendix~\ref{app:CBFn}, the structure of the signs appearing in the stabilizer group of the $n$-qubit cycle-graph state $|C_n\rangle$ always has the form of a cubic Boolean function $g_n(\mathbf{x}) = \sum_{j=0}^{n-1} x_{j-1}x_jx_{j+1}$.  Thus the quantum strategy where each party $j$ measures the Pauli observable $E_j(a_j,b_j)$ is always perfect.  As for classical strategies, the game $\mathrm{CBF}(C_n,\{0,1\}^n)$ for $n$ parties always exhibits a depth-0 bound $\beta_0<1$ (cf.~Theorem 1 of Ref.~\cite{guhne2005bell}).  However, the exact value of $\beta_0$ for $n>10$ is unknown and cannot be obtained efficiently numerically. 
A general discussion of cubic Boolean function games defined on arbitrary graphs will be presented in future work.
On the other hand, this game exhibits a trivial depth-1 bound $\beta_1=1$, as the strategy whereby each party outputs $y_j = a_{j-1}a_j(b_j+a_{j-1}) = x_{j-1}x_jx_{j+1}~\mathrm{mod}~2$ is perfect.  Thus, the stabilizer submeasurement games cannot be used to demonstrate quantum advantage in terms of circuit depth.

\subsection{$C_n$ stabilizer submeasurement games}

The premise and win conditions of the $C_6$ stabilizer submeasurement games introduced in Sec.~\ref{sec:ss_games} generalize to the $n$-qubit scenario as follows.

\begin{def:general_Stab_Sub_Game}[$\mathrm{SS}(C_n,\mathcal{I})$]
Consider a game where $n$ parties corresponding to the vertices of the cycle graph $C_n$ are each given a one-bit input $x_j$, $j=0,\ldots,n-1$, from a global input $\mathbf{x}=(x_0,\ldots,x_{n-1})\in\mathcal{I}\subseteq\{0,1\}^n$ drawn at random from the uniform distribution. 
To each input we can associate a global Pauli operator $E(\mathbf{1},\mathbf{x}) = \prod_{j=0}^{{n-1}} i^{x_j} X_j Z_j^{x_j}$, where $\mathbf{1}$ denotes the all-ones string.  Each party then produces an output $y_j\in\{0,1\}$, which collectively forms a string $\mathbf{y}\in\{0,1\}^n$.  The parties are said to win the game whenever
\begin{subequations}\label{eq:SS_win_conditions}
\begin{align}
&\forall P\subseteq E(\mathbf{1},\mathbf{x}) ~\mathrm{s.t.}~P\in\mathcal{S}_{C_n},~\sum_{j\in\mathrm{supp}(P)} y_j = 0 ,\\
&\forall P\subseteq E(\mathbf{1},\mathbf{x}) ~\mathrm{s.t.}~P\in-\mathcal{S}_{C_n},~\sum_{j\in\mathrm{supp}(P)} y_j = 1 ,
\end{align}
\end{subequations}
where $P\subseteq E(\mathbf{1},\mathbf{x})$ means that $P$ is a Pauli operator obtained by replacing some nonidentity tensor factors in $E(\mathbf{1},\mathbf{x})$ with the identity.  Furthermore, $\mathrm{supp}(P) = \mathrm{supp}(\mathbf{a},\mathbf{b})$ where $P = E(\mathbf{a},\mathbf{b})$ (i.e., $\mathrm{supp}(P)$ indexes the qubits where $P$ acts nontrivially) and all arithmetic is performed modulo 2.
\end{def:general_Stab_Sub_Game}

A perfect quantum strategy for this game can be performed by preparing the state $|C_n\rangle$ with a constant-depth quantum circuit and measuring each Pauli term in $E(\mathbf{1},\mathbf{x})$.  On the other hand, geometrically restricted classical circuits require depth $\Omega(n)$ to attain perfect strategies.  Let
\begin{equation}\label{eq:HLF_n_5}
\mathcal{I}_{\mathrm{HLF},n}^{(5)} = \l\{\mathbf{x}\in\{0,1\}^n \middle\vert
\begin{array}{l}
(x_0, x_{2\lfloor  n/6 \rfloor}, x_{2\lfloor n/3 \rfloor}) \in \mathcal{V}, \\
x_j = 0 \text{ for all other } j
\end{array}
\r\},
\end{equation}
where $\mathcal{V} \subset \{0,1\}^3$ is defined as
\begin{equation}
    \mathcal{V} = \{(0,0,0), (0,1,1), (1,0,1), (1,1,0), (1,1,1)\}.
\end{equation}
The following theorem then holds.
\begin{theorem}
\label{thm:n_qubit_ss_games}
Let $D$ be odd and let $n=6D$.  No classical Boolean circuit with depth $D$ and fan-in~$\leq 3$ that is geometrically restricted with respect to the cycle graph $C_n$ can win $\mathrm{SS}(C_n,\mathcal{I}_{\mathrm{HLF},n}^{(5)})$ with average success probability greater than $\beta_D=4/5$. Meanwhile, a constant-depth quantum circuit with the same geometry can achieve $\mathrm{Pr}_Q[\mathrm{win}]=1$.
\end{theorem}
We relegate the proof of this statement to Appendix~\ref{app:Depth_D_Bounds}.  The difference in the asymptotic scaling in the classical and quantum circuits depths required to achieve a perfect strategy implies the existence of a quantum computational advantage. Phrased differently, the above theorem demonstrates an unconditional separation between the power of constant-depth quantum circuits and sublinear-depth classical circuits with the same cyclic geometry. 

\begin{remark}
For every even number $6(2m-1)\leq n<6(2m+1)$ the game $\mathrm{SS}(C_n,\mathcal{I}_{\mathrm{HLF},n}^{(5)})$ satisfies the bounds $\beta_{2m-1}=4/5$ and $\beta_{2m}=1$.  Hence, for each odd $D$, $n=6D$ is the smallest example giving the next nontrivial bound of $\beta_D=4/5$.
\end{remark}

 {That is why the the next nontrivial sized instance of this game is for $n=18$.  In this case, a constant-depth quantum circuit can in principle exceed the success probability bound on depth-3 classical circuits of the same cyclic geometry.}
 
Finally, in order to demonstrate quantum advantage against the more powerful class of constant-depth classical circuits without geometric restriction (namely those in Nick's class $\mathsf{NC}^0$), it is imperative either to embed the stabilizer submeasurement game into a 2D grid~\cite{bravyi2018quantum}, or instead play a more complex game such as the magic square game~\cite{bravyi2020quantum}, which is classically hard even on a 1D geometry.  Note that small instances of the 1D magic square game incur a constant overhead in the number of qubits and entangling gates required, and they have a higher classical bound of $8/9$, in contrast to the games proposed here.

\section{\label{sec:discussion}Conclusion} 

In this work, we have proposed and implemented proof-of-principle experimental demonstrations of two types of nonlocal games using six qubits on a gate-based trapped-ion quantum computer.  The second family of games, called the stabilizer submeasurement games, are computational tasks that unconditionally prove quantum advantage against sublinear-depth classical circuits with a geometric restriction in their gate connectivity to the cycle graph.  As summarized in Table~\ref{tab:experimental_results}, our NISQ device surpasses the conventional depth-0 bounds for the first family of games, called cubic Boolean function games, by a significant margin.  These games also provide fine-grained details like state fidelity, which is useful for the characterization of noise in a quantum device. On the other hand,
our results suggest that state-of-the-art devices are now at the level of challenging the more difficult depth-1 bounds of the stabilizer submeasurement games.


As discussed in Sec.~\ref{sec:generalization}, both games generalize nicely to the $n$-qubit scenario. Then, demonstration of an advantageous quantum strategy for the $C_n$ stabilizer submeasurement games gives a particularly meaningful separation against geometrically-restricted classical circuits of depth growing sublinearly with the size of the input according to Theorem~\ref{thm:n_qubit_ss_games}. For example, the stabilizer submeasurement game on an 18-qubit device provides a next target,  {in that surpassing the corresponding bounds would demonstrate achievement of a task that depth-3 classical circuits with a cyclic geometry cannot}. Furthermore, these nonlocal games are friendly to a number of experimental platforms to implement quantum simulation of many-body physics, as these advantageous quantum strategies on a cyclic cluster state could be generalized to 
generic ground states of  {many-body Hamiltonians with} 1D symmetry-protected topological order~\cite{daniel2021quantum}.


\textit{Note added}.~After the completion of the work, it is brought to our attention that a new version of Ref.~\cite{sheffer2021playing} reports, without detailed data, the Honeywell H1 processor \cite{pino2021demonstration} achieves a success probability of 96.9(3)\% for the 2D hidden linear function game in Sec.~\ref{sec:SS8}, which surpasses the depth-1 classical bound.



\begin{acknowledgments}
    We thank Kenneth Brown for helpful discussions and the STAQ Project for providing an opportunity of the collaboration. We also thank the UNM Center for Advanced Research Computing, supported in part by the National Science Foundation, for providing the high-performance computing resources used in this work. The work at UNM led by A.M.~was supported partially by the National Science Foundation STAQ Project (PHY-1818914), PHY-1915011, and the Department of Energy, Office of Science National Quantum Information Science Research Center, Quantum Systems Accelerator. N.M.L.~acknowledges support from the NSF Physics Frontier Center (PHY-1430094) at the Joint Quantum Institute (JQI), the Maryland-Army-Research-Lab Quantum Partnership (W911NF1920181), and the Office of Naval Research (N00014-20-1-2695). A.M.G.~is supported by a JQI Postdoctoral Fellowship.

\end{acknowledgments}



\appendix

\section{\label{sec:stab_sub_details}Additional details on stabilizer submeasurements and expected populations}

In this appendix we elaborate on the notion of stabilizer submeasurements used in the main text.  We discuss the measurement outcomes expected when a stabilizer exists as a subset of the observables measured in a global Pauli string, which is indicative of the quantum correlations between those outcomes.

The following lemma regarding the stabilizer formalism is well known~\cite[Sec.~10.5.3]{nielsen_chuang_2010}.

\begin{lemma}
The expectation value of any $n$-qubit Pauli operator over an $n$-qubit stabilizer state $|\psi\rangle$ with stabilizer group $\mathcal{S}$ is
\begin{subequations}
\begin{align}
    \langle \psi | P | \psi \rangle = +1 \quad &\textrm{ iff }P\in\mathcal{S}\\
    \langle \psi | P | \psi \rangle = -1  \quad &\textrm{ iff }P\in-\mathcal{S}\\
    \langle \psi | P | \psi \rangle = 0  \quad &\textrm{ iff }P\not\in\mathcal{S}.
\end{align}
\end{subequations}
\end{lemma}

The following is an extension of this lemma for the case when $P$ contains stabilizers within it.
\begin{lemma}
Consider a global $n$-qubit Pauli operator $P = \prod_{j=1}^n P_j$ for some choice of $P_j\in\{X_j,Y_j,Z_j\}$ (i.e., a Pauli string $P$ that has nontrivial support on every qubit).  Let $P\cap \mathcal{S}\subseteq \mathcal{S}$ denote stabilizer elements that are contained in $P$.  Denote each such $S\in P\cap\mathcal{S}$ with a binary string $\mathbf{f}^{(S)}\in\{0,1\}^n$ whereby $S = (-1)^{\lambda_S}\prod_{j=1}^n P_j^{f_j^{(S)}}$ with $\lambda_S\in\{0,1\}$. If a measurement of $P$ is performed by locally measuring each Pauli operator in the string, then the outcomes $\mathbf{y}\in\{0,1\}^n$ occur with probability
\begin{equation}
\Pr[\mathbf{y}] =
\begin{cases}
|P\cap\mathcal{S}| / 2^{n} & \text{if } \mathbf{f}^{(S)}\cdot\mathbf{y}=\lambda_S \text{ for all } S\in P\cap\mathcal{S} \\
0 & \text{ otherwise.}
\end{cases}
\end{equation}
\label{lemma:Stab_Sub_Probs}
\end{lemma}
\begin{proof}
Suppose that $U$ is the local Clifford unitary that diagonalizes $P$ in the sense that $UPU^\dagger = \prod_{j=1}^n Z_j$.  Then for each $S\in P\cap\mathcal{S}$, $USU^\dagger = (-1)^{\lambda_{S}}\prod_{j=1}^n Z_j^{f_j^{(S)}}$.  First note that
\begin{equation}
    \begin{split}
    \langle\mathbf{y}|U|\psi\rangle &= \langle \mathbf{y}|U S|\psi\rangle \\
    &= \langle \mathbf{y}|U SU^\dagger U|\psi\rangle \\
    &= \langle \mathbf{y}|(-1)^{\lambda_S} \prod_{j=1}^n Z_j^{f_j^{(S)}} U|\psi\rangle \\
    &=  (-1)^{\mathbf{f}^{(S)}\cdot \mathbf{y} + \lambda_S}\langle \mathbf{y}| U|\psi\rangle,
    \end{split}
\end{equation}
which implies that either $\lambda_S = \mathbf{f}^{(S)} \cdot \mathbf{y}$ or $|\langle\mathbf{y}|U|\psi\rangle|^2 = \Pr[\mathbf{y}] = 0$.  

On the other hand, we have that
\begin{equation}
    \begin{split}
    |\langle \mathbf{y}|U|\psi\rangle |^2 &= \langle \mathbf{y} | \frac{1}{2^n}\sum_{S\in\mathcal{S}} U S U^\dagger|\mathbf{y}\rangle \\
    &= \frac{1}{2^n}\sum_{S\in P\cap\mathcal{S}} (-1)^{\lambda_S} \langle\mathbf{y}|\prod_{j=1}^n Z_j^{f^{(S)}_j}|\mathbf{y}\rangle \\
    &= \frac{1}{2^n}\sum_{S\in P\cap\mathcal{S}} (-1)^{\lambda_S}  (-1)^{\mathbf{f}^{(S)}\cdot \mathbf{y}} \\
    &= \frac{1}{2^n}\sum_{S\in P\cap\mathcal{S}} (+1) \\
    &= \frac{|P\cap\mathcal{S}|}{2^{n}}.
    \end{split}
\end{equation}
Therefore, the measurement outcomes $\mathbf{y}\in\{0,1\}^n$ that occur with nonzero probability are those satisfying $\mathbf{f}^{(S)}\cdot\mathbf{y}=\lambda_S$ $\forall S\in P\cap\mathcal{S}$.  Furthermore, each such outcome occurs with uniform probability $|P\cap\mathcal{S}|/2^{n}$.
\end{proof}



\section{Bell operators for stabilizer submeasurement games}

Here we prove a useful expression for computing the average success probability for the stabilizer submeasurement games in terms of the average output parities for each stabilizer submeasurement $P \subseteq E(\mathbf{1}, \mathbf{x})$, which we shall denote by $\langle P \rangle$. This result applies for both quantum and classical strategies:~for a quantum strategy using the state $\rho$,
\begin{equation}
\langle P \rangle = \tr(\rho P),
\end{equation}
while for a classical deterministic strategy assigning outputs $\mathbf{y} \in \{0, 1\}^n$ to each observable in $P$ depending on the input,
\begin{equation}
    \langle P \rangle = (-1)^{\sum_{j \in \mathrm{supp}(P)} y_j}.
\end{equation}

\begin{lemma}
\label{thm:Bell_Operator}
Consider a strategy for any stabilizer submeasurement game $\mathrm{SS}(C_n,\mathcal{I})$. The average success probability can be written in terms of expectation values of the stabilizing operators,
\begin{align}\label{eq:win_prob_stabilizers}
    \Pr[\mathrm{win}] = \frac{1}{|\mathcal{I}|}\sum_{\mathbf{x}\in\mathcal{I}}\frac{1}{|P_\mathbf{x}|}\sum_{P\in P_\mathbf{x}} \sgn(P) \langle P \rangle,
\end{align}
where $P_\mathbf{x} = \{P\subseteq E(\mathbf{1},\mathbf{x}) \mid P\in\pm\mathcal{S}_{C_n}\}$ and $\sgn(P) = \pm1$ whenever $P\in\pm\mathcal{S}_{C_n}$
\end{lemma}

\begin{proof}
Recall from Eq.~\eqref{eq:success_prob_def} that the average success probability can be computed as
\begin{align}
    \Pr[\mathrm{win}] = \frac{1}{|\mathcal{I}|} \sum_{\mathbf{x}\in\mathcal{I}}\sum_{\mathbf{y}\in\mathrm{Win}(\mathbf{x})} \Pr[\mathbf{y}|\mathbf{x}],
\end{align}
where for the game $\mathrm{SS}(C_n,\mathcal{I})$ we have
\begin{widetext}
\begin{equation}
    \mathrm{Win}(\mathbf{x}) =
    \left\{ \mathbf{y}\in\{0,1\}^n ~ \middle\vert ~ \forall P\in P_\mathbf{x},~(-1)^{\sum_{j\in\mathrm{supp}(P)}y_j} = \sgn(P) \right\}.
\end{equation}
Define the single-qubit Weyl operators as $E(a,b) = i^{ab} X^a Z^b$ for $a,b\in\{0,1\}$. The quantity $\Pr[\mathbf{y}|\mathbf{x}]$ can be computed as
\begin{align}
    \Pr[\mathbf{y}|\mathbf{x}] = \left\langle \bigotimes_{j=0}^{n-1} \Pi_{y_j}^{(E(1,x_j))}\right\rangle ,
\end{align}
where $\Pi_{y_j}^{(E(1,x_j))}$ denotes the projector onto the eigenspace corresponding to eigenvalue $(-1)^{y_j}$ of the operator $E(1,x_j)$. Since $E(1,x_j)$ is a single-qubit Pauli observable, we can write
\begin{align}
    \Pi_{y_j}^{(E(1,x_j))} = \frac{I+(-1)^{y_j}E(1,x_j)}{2},
\end{align}
hence
\begin{align}
    \Pr[\mathrm{win}] = \frac{1}{|\mathcal{I}|} \sum_{\mathbf{x}\in\mathcal{I}}\sum_{\mathbf{y}\in\mathrm{Win}(\mathbf{x})} \left\langle\bigotimes_{j=0}^{n-1} \frac{I+(-1)^{y_j}E(1,x_j)}{2} \right\rangle.
\end{align}
Upon binomial-expanding the product in the expectation value, we get
\begin{align}
    \Pr[\mathrm{win}] &= \frac{1}{2^n|\mathcal{I}|} \sum_{\mathbf{x}\in\mathcal{I}}\sum_{\mathbf{y}\in\mathrm{Win}(\mathbf{x})} \sum_{\mathbf{q}\in\{0,1\}^n}\left\langle\bigotimes_{j=0}^{n-1} (-1)^{q_j y_j} E(1,x_j)^{q_j} \right\rangle \\
    &=\frac{1}{2^n|\mathcal{I}|} \sum_{\mathbf{x}\in\mathcal{I}}\sum_{\mathbf{y}\in\mathrm{Win}(\mathbf{x})} \sum_{\mathbf{q}\in\{0,1\}^n}(-1)^{\mathbf{q}\cdot\mathbf{y}}\left\langle E(\mathbf{q},\mathbf{q}\odot\mathbf{x}) \right\rangle,
\end{align}
where $\mathbf{q} \odot \mathbf{x} = (q_0x_0,\ldots,q_{n-1}x_{n-1})$ denotes elementwise multiplication of vectors.  Notice that the $j$th term in the $n$-qubit Pauli operator $E(\mathbf{q},\mathbf{q}\odot\mathbf{x})$ is $I$ if $q_j=0$, and $E(1,x_j)$ if $q_j=1$.  Therefore, the sum over $\mathbf{q}$ is the same as a sum over all Pauli operators $P$ such that $P\subseteq E(\mathbf{1},\mathbf{x})$.  Hence,
\begin{align}
    \Pr[\mathrm{win}] &=\frac{1}{2^n|\mathcal{I}|} \sum_{\mathbf{x}\in\mathcal{I}}\sum_{\mathbf{y}\in\mathrm{Win}(\mathbf{x})} \sum_{P\subseteq E(\mathbf{1},\mathbf{x})}(-1)^{\sum_{j\in\mathrm{supp}(P)}y_j}\left\langle P \right\rangle \\
    &= \frac{1}{2^n|\mathcal{I}|} \sum_{\mathbf{x}\in\mathcal{I}} \sum_{P\subseteq E(\mathbf{1},\mathbf{x})}\sum_{\mathbf{y}\in\mathrm{Win}(\mathbf{x})}(-1)^{\sum_{j\in\mathrm{supp}(P)}y_j}\left\langle P \right\rangle.
\end{align}
If $P\in P_\mathbf{x}$, then $\sum_{\mathbf{y}\in\mathrm{Win}(\mathbf{x})}(-1)^{\sum_{j\in\mathrm{supp}(P)}y_j} = |\mathrm{Win}(\mathbf{x})|\sgn(P)$ by definition of $\mathrm{Win}(\mathbf{x})$.  On the other hand, if $P\not\in P_\mathbf{x}$, then $\sum_{\mathbf{y}\in\mathrm{Win}(\mathbf{x})}(-1)^{\sum_{j\in\mathrm{supp}(P)}y_j} = 0$.  Intuitively, this can be understood from the fact that if $P\not\in\pm\mathcal{S}$, then $\sum_{j\in\mathrm{supp}(P)}y_j$ is not constrained by a win condition.  Hence, this quantity will be unconstrained and the corresponding sum over $\mathbf{y}\in\mathrm{Win}(\mathbf{x})$ will vanish.

More rigorously, this occurs because for each $\mathbf{x}\in\mathcal{I}$, $\mathrm{Win}(\mathbf{x})$ corresponds to the solution set of an inhomogeneous system of linear equations $A \mathbf{y} = \mathbf{b}$ over $\mathbb{F}_2$, where $A$ is a matrix where each row $\mu$ corresponds to a binary vector $\mathbf{a}_\mu$ such that $E(\mathbf{a}_\mu,\mathbf{a}_{\mu}\odot\mathbf{x})\in\pm\mathcal{S}$ and $(-1)^{b_\mu} = \sgn(E(\mathbf{a}_\mu,\mathbf{a}_\mu\odot\mathbf{x}))$.  Denote the row space of this matrix as $\mathrm{Row}(A)$.  It is then clear that $\mathrm{rank}(A)=\log_2|P_\mathbf{x}|$, hence $ | \mathrm{Win}(\mathbf{x}) | = 2^{n}/|P_\mathbf{x}|$ by the rank--nullity theorem. Furthermore, each $\mathbf{y}\in\mathrm{Win}(\mathbf{x})$ can be written as $\mathbf{y} = \mathbf{y}_0 + \sum_{\nu=1}^{n-\mathrm{rank}(A)} z_\nu \boldsymbol{\eta}_\nu$, where $A\mathbf{y}_0 = \mathbf{b}$, $z_\nu\in\mathbb{F}_2$, and $\{\boldsymbol{\eta}_\nu\}_{\nu=1}^{n-\mathrm{rank}(A)}$ spans the null space of $A$.  It then follows that for any $P = E(\mathbf{q},\mathbf{q}\odot\mathbf{x})\not\in\pm\mathcal{S}$,
\begin{align}
    \sum_{\mathbf{y}\in\mathrm{Win}(\mathbf{x})}(-1)^{\sum_{j\in\mathrm{supp}(P)}y_j} = (-1)^{\mathbf{q}\cdot\mathbf{y}_0}\prod_{\nu=1}^{n-\mathrm{rank}(A)}\sum_{z_\nu=0}^1 (-1)^{z_\nu \mathbf{q}\cdot\boldsymbol{\eta}_\nu}.
\end{align}
The only way for the right-hand side to be nonzero is for $\mathbf{q}\cdot\boldsymbol{\eta}_\nu=0$ for each $\nu=1,\ldots,n-\mathrm{rank}(A)$;~however, this can only be true if $\mathbf{q}\in\mathrm{Row}(A)$, which does not hold by virtue of the fact that $P\not\in\pm\mathcal{S}$.  Therefore
\begin{align}
    \Pr[\mathrm{win}] = \frac{1}{|\mathcal{I}|}\sum_{\mathbf{x}\in\mathcal{I}}\frac{1}{|P_\mathbf{x}|}\sum_{P\in P_\mathbf{x}} \sgn(P)\langle P \rangle,
\end{align}
\end{widetext}
which completes the proof.
\end{proof}

Using the above expression, for the game $\mathrm{SS}(C_6,\mathcal{I}_{\mathrm{HLF}}^{(8)})$ we obtain
\begin{align}
    \Pr[\mathrm{win}]= &\frac{1}{32}\big(12 + 12 \langle X_1 X_3 X_5 \rangle + \langle X_0X_2X_4 \rangle   \nonumber\\
    &+ \langle X_0 X_1 X_2 X_3 X_4 X_5 \rangle \nonumber \\
&- \langle X_0 Y_2 X_3 Y_4\rangle - \langle X_0 X_1 Y_2 Y_4 X_5\rangle \nonumber \\
& - \langle Y_0 X_2 Y_4 X_5\rangle - \langle Y_0 X_1 X_2 X_3 Y_4\rangle \nonumber \\
&- \langle Y_0 X_1 Y_2 X_4\rangle - \langle Y_0 Y_2 X_3 X_4 X_5\rangle\big) .
\end{align}
Similarly, for the game $\mathrm{SS}(C_6,\mathcal{I}_{\mathrm{HLF}}^{(5)})$ we obtain
\begin{align}
    \Pr[\mathrm{win}] = &\frac{1}{20}\big(6 + 6 \langle X_1 X_3 X_5 \rangle + \langle X_0X_2X_4 \rangle  \nonumber \\
    &+ \langle X_0 X_1 X_2 X_3 X_4 X_5 \rangle \nonumber \\
&- \langle X_0 Y_2 X_3 Y_4\rangle - \langle X_0 X_1 Y_2 Y_4 X_5\rangle \nonumber \\
& - \langle Y_0 X_2 Y_4 X_5\rangle - \langle Y_0 X_1 X_2 X_3 Y_4\rangle \nonumber \\
&- \langle Y_0 X_1 Y_2 X_4\rangle - \langle Y_0 Y_2 X_3 X_4 X_5\rangle\big).
\end{align}



\section{Quantum strategy for the $C_6$ cubic Boolean function game}
\label{app:CBFn}

To demonstrate that the quantum strategy for the $C_6$ cubic Boolean function game is perfect, we prove a lemma regarding the structure of the signs appearing in the stabilizer group of the $n$-qubit cyclic cluster state.  First, define the single-qubit Weyl operators as $E(a,b) = i^{ab} X^a Z^b$ for $a,b\in\{0,1\}$.  Furthermore, denote the generators of the $n$-qubit cyclic cluster state stabilizer group as $S_j=Z_{j-1}X_{j}Z_{j+1}$ where the subscripts are taken mod~$n$.  We then have the following lemma.

\begin{lemma}
\label{lemma:CBF_n}
Every element of the stabilizer group of the $n$-qubit cyclic cluster state $|C_n\rangle$, when denoted as a product of generators $S_\mathbf{x}=\prod_{j=0}^{n-1} S_j^{x_j}$ for some $\mathbf{x}\in\{0,1\}^n$, can be expressed as a Pauli string times a phase via
\begin{align}
    S_{\mathbf{x}} = (-1)^{g_n(\mathbf{x})} \bigotimes_{j=0}^{n-1} E(x_j,x_{j-1} + x_{j+1}),
\end{align}
where $g_n(\mathbf{x})=\sum_{j= 0}^{n-1} x_{j-1} x_j x_{j+1}$ and all arithmetic is performed $ \mathrm{mod}~2 $.
\end{lemma}

\begin{proof}
For the $n$-qubit cycle graph state,
\begin{align}
S_\mathbf{x} = \prod_{j=0}^{n-1} (Z_{j-1} X_j Z_{j+1})^{x_j} = \bigotimes_{j=0}^{n-1} Z^{x_{j-1}} X^{x_j} Z^{x_{j+1}}
\end{align}
Each local Pauli operator can be expressed as a Weyl operator times a phase. In terms of the Weyl operators, this expression becomes
\begin{align}
S_\mathbf{x} &=\bigotimes_{j=0}^{n-1} (-1)^{x_{j-1}x_j} X^{x_j} Z^{x_{j-1}+x_{j+1}} \nonumber\\
&= \bigotimes_{j=0}^{n-1} (-1)^{x_{j}x_{j+1}}i^{x_j(x_{j-1}\oplus x_{j+1})} E(x_j,x_{j-1}+x_{j+1}) \nonumber\\
&= (-1)^{\sum_{j=0}^{n-1} x_j x_{j+1}}\prod_{j=0}^{n-1} i^{x_j(x_{j-1}\oplus x_{j+1})} \nonumber\\
&\qquad \times \bigotimes_{k=0}^{n-1} E(x_k, x_{k-1}+x_{k+1}).
\end{align}
Here we have used $\oplus$ to denote the XOR operation (i.e., sum mod 2) in the exponent of $i$ to remind the reader that the exponent must be computed mod 2.  The phase appearing in the last line then has the rather unnatural form of $i$ raised to the power of a Boolean function.  A simple analysis shows that the total number of factors of $i$ accumulated is $\sum_{j=0}^{n-1} x_j(x_{j-1}\oplus x_{j+1})$, which is simply twice the number of connected strings of $1$'s in $\mathbf{x}$ with length greater than one. We thus obtain an overall phase of $\pm1$ whenever this number of connected strings equals 0 or 2 mod 4, respectively.  This phase is also expressed by the function $\sum_{j=0}^{n-1} (x_{j-1}+1)x_{j}x_{j+1}$, but using a base of $-1$ instead of $i$.  Thus we have
\begin{align}
S_\mathbf{x} = (-1)^{g_n(\mathbf{x})} \bigotimes_{k=0}^{n-1} E(x_k,x_{k-1}+x_{k+1}),
\end{align}
which is what we sought to prove.
\end{proof}
It is clear that since $S_\mathbf{x}|C_n\rangle = |C_n\rangle$ we have that,
\begin{align}
    \langle C_n|\bigotimes_{j=0}^{n-1} E(x_j,x_{j-1}+x_{j+1})|C_n\rangle = (-1)^{g_n(\mathbf{x})}.
\end{align}
Therefore, locally measuring each Pauli observable in the string $\bigotimes_{j=0}^{n-1} E(x_j,x_{j-1}+x_{j+1})$ gives measurement outcomes whose parity is equal to $g_n(\mathbf{x})\textrm{ mod 2}$. Hence the strategy presented in the text for the $C_6$ cubic Boolean function game is perfect.

\section{Classical strategies for stabilizer submeasurement games}

Here we derive the classical bounds for the stabilizer submeasurement games, assuming the classical strategy is restricted to depth-0, depth-1, or depth-$ D $ circuits. In particular, we show that for the particular input set $ \mathcal{I}^{(5)}_{\mathrm{HLF}, n} $ (defined in Eq.~\eqref{eq:HLF_n_5} in the main text), the classical bounds are $ 4/5 $ in all cases considered. This improves upon the previously known bound of $ 7/8 $, which corresponds to the input set $ \mathcal{I}^{(8)}_{\mathrm{HLF}, n} $~\cite{bravyi2018quantum}.

\subsection{Depth-0 bounds for stabilizer submeasurement games}

\begin{lemma}
\label{lemma:Depth0}
The depth-$0$ bound for any of the $\mathrm{SS}(C_n,\mathcal{I})$ games can be calculated using the expression in Lemma~\ref{thm:Bell_Operator}.  It follows that no depth-$0$ circuit can win the game $\mathrm{SS}(C_6,\mathcal{I}_{\mathrm{HLF}}^{(8)})$ with probability greater than $7/8$, nor can they win the game $\mathrm{SS}(C_6,\mathcal{I}_{\mathrm{HLF}}^{(5)})$ with probability greater than $4/5$.
\end{lemma}

\begin{proof}
Our proof gives a general technique to calculate this bound for any of the $\mathrm{SS}(C_n,\mathcal{I})$ games.  As a classical circuit, depth-0 strategies for these games correspond to circuits with one layer of single-input, single-output gates each computing an affine Boolean function $y_j(x_j)$.  The average success probability can then be calculated from the Bell expression given in Lemma~\ref{thm:Bell_Operator} as
\begin{widetext}
\begin{align}
    2\Pr[\mathrm{win}] - 1 = \frac{1}{|\mathcal{I}|}\sum_{\mathbf{x}\in\mathcal{I}} \left\lfloor \frac{1}{|P_\mathbf{x}|} \sum_{P\in P_\mathbf{x}} (-1)^{\sum_{j\in \mathrm{supp}(P)} y_j(x_j)}\mathrm{sgn}(P) \right\rfloor
\end{align}
\end{widetext}
The floor function takes care of the fact that if one or more on the parity constraints imposed by the win conditions is not met, then the round is counted as a failure [cf.~Eq.~\eqref{eq:SS_win_conditions}].  By maximizing over all possible local functions $y_j:\{0,1\}\rightarrow \{0,1\}$ we can numerically find these bounds for $ \mathrm{SS}(C_6, \mathcal{I}_{\mathrm{HLF}}^{(8)}) $ and $ \mathrm{SS}(C_6, \mathcal{I}_{\mathrm{HLF}}^{(5)}) $ via brute-force search.
\end{proof}
Note that Lemma~\ref{lemma:Depth0} could also be proven by the proof technique in the following Lemma~\ref{lemma:Depth1}.

\subsection{Depth-1 bounds for $\mathrm{SS}(C_6,\mathcal{I}_{\mathrm{HLF}}^{(8)})$ and $\mathrm{SS}(C_6,\mathcal{I}_{\mathrm{HLF}}^{(5)})$}
\label{sec:Triangle_Game_Bounds}
\begin{lemma}
\label{lemma:Depth1}
Depth-$1$, fan-in~$\leq3$ classical circuits that are geometrically restricted with respect to the geometry of $C_6$ cannot win the game $\mathrm{SS}(C_6,\mathcal{I}_{\mathrm{HLF}}^{(8)})$ on more than $7/8$ of the inputs, nor can they win the game $\mathrm{SS}(C_6,\mathcal{I}_{\mathrm{HLF}}^{(5)})$ on more than $4/5$ of the inputs.
\end{lemma}

\begin{figure}[h]
    \centering
    \includegraphics[width=0.2\linewidth]{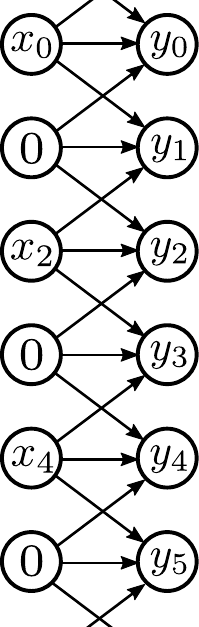}
    \caption{A general depth-1 classical circuit with 3-local conditional gates, conditioned on their nearest-neighboring inputs.  Note that the action of such a gate on one of the even bits is equivalent to a local computation (i.e., their neighboring inputs are fixed as 0).}
    \label{fig:Fanin3}
\end{figure}

\begin{proof}
In the following proof all arithmetic is performed modulo 2.  Consider a classical circuit of depth-1 and fan-in~$\leq3$ that is geometrically restricted with respect to the geometry of $C_6$.  When given an input from $\mathcal{I}_{\mathrm{HLF}}^{(8)}$ or $\mathcal{I}_{\mathrm{HLF}}^{(5)}$, as depicted in Fig.~\ref{fig:Fanin3}, the circuit produces outputs that can be parameterized as follows:
\begin{subequations}
\begin{align}
    y_0 &= \alpha_0 + \beta_0 x_0 \label{eq:y0}\\
    y_1 &= \alpha_1 + \beta_1 x_0 + \gamma_1 x_2 + \delta_1 x_0 x_2 \label{eq:y1}\\
    y_2 &= \alpha_2 + \beta_2 x_2 \label{eq:y2}\\
    y_3 &= \alpha_3 + \beta_3 x_2 + \gamma_3 x_4 + \delta_3 x_2 x_4 \label{eq:y3}\\
    y_4 &= \alpha_4 + \beta_4 x_4 \label{eq:y4}\\
    y_5 &= \alpha_5 + \beta_5 x_4 + \gamma_5 x_0 + \delta_5 x_0 x_4 ,\label{eq:y5}
\end{align}
\end{subequations}
where $\alpha_j,\beta_j,\gamma_j,\delta_j\in\{0,1\}$ for $j=0,\ldots,5$.

If this strategy is to win the game $\mathrm{SS}(C_6,\mathcal{I}_{\mathrm{HLF}}^{(8)})$, it must hold that $y_1 + y_3 + y_5 = 0$ for every input. In particular, for the game $\mathrm{SS}(C_6,\mathcal{I}_{\mathrm{HLF}}^{(5)})$ this must hold for the input $\mathbf{x}=101010$.  By inserting Eqs.~\eqref{eq:y1}, \eqref{eq:y3}, and \eqref{eq:y5} into this expression for $\mathbf{x}=101010$, we obtain a constraint on $\alpha_j$, $\beta_j$, $\gamma_j$, and $\delta_j$ for each $j\in\{1,3,5\}$.  Namely,
\begin{align}
    \sum_{j\in\{1,3,5\}} \alpha_j + \beta_j + \gamma_j + \delta_j = 0. \label{eq:111constraint}
\end{align}

Let $e(\mathbf{x}) = y_0 + y_2 + y_4 = \alpha + \beta_0 x_0 + \beta_2 x_2 + \beta_4 x_4$.  Here we have defined $\alpha = \alpha_0 + \alpha_2 + \alpha_4$.  The win conditions for inputs $\mathbf{x}\in\mathcal{I}_{\mathrm{HLF}}^{(5)}\setminus\{101010\}$ imply that the following linear constraints must hold:
\begin{subequations}
\begin{align}
    e(000000) &= 0 \\
    e(101000) + y_1(101000) &= 1 \\
    e(001010) + y_3(001010) &= 1 \\
    e(100010) + y_5(100010) &= 1.
\end{align}
\end{subequations}
Inserting the aforementioned Boolean functions into these expressions gives
\begin{subequations}
\begin{align}
    \alpha &= 0 \label{eq:000contraint}\\
    \alpha + \beta_0 + \beta_2 + \alpha_1 + \beta_1 + \gamma_1 + \delta_1 &= 1 \\
    \alpha + \beta_2 + \beta_4 + \alpha_3 + \beta_3 + \gamma_3 + \delta_3 &= 1 \\
    \alpha + \beta_0 + \beta_4 + \alpha_5 + \beta_5 + \gamma_5 + \delta_5 &= 1. \label{eq:101constraint}
\end{align}
\end{subequations}
Summing Eq.~\eqref{eq:111constraint} and Eqs.~\eqref{eq:000contraint}--\eqref{eq:101constraint} modulo~$2$ gives the contradictory statement that $0=1$.  Therefore, this class of circuits wins the game $\mathrm{SS}(C_6,\mathcal{I}_{\mathrm{HLF}}^{(8)})$ on no more than $7/8$ of the inputs, and the game $\mathrm{SS}(C_6,\mathcal{I}_{\mathrm{HLF}}^{(5)})$ on no more than $4/5$ of the inputs.
\end{proof}

\subsection{Depth-$D$ bounds for $\mathrm{SS}(C_{6D},\mathcal{I}_{\mathrm{HLF},{6D}}^{(5)})$}
\label{app:Depth_D_Bounds}

\begin{reptheorem}{thm:n_qubit_ss_games}[Restated from main text.]
Let $D$ be odd and let $n=6D$.  No classical circuit with depth-$D$ and fan-in~$\leq 3$ that is geometrically restricted with respect to the cycle graph $C_n$ can win $\mathrm{SS}(C_n,\mathcal{I}_{\mathrm{HLF},n}^{(5)})$ with average success probability greater than $\beta_D=4/5$. Meanwhile, a constant-depth quantum circuit with the same geometry can achieve $\mathrm{Pr}_Q[\mathrm{win}]=1$.
\end{reptheorem}

\begin{proof}
The proof follows similarly to the six-qubit case. All arithmetic here is performed modulo 2.  Let $D$ be an odd positive integer.  Let $\mathcal{I}_{\mathrm{HLF},{6D}}^{(5)} = \{\mathbf{x}\in\{0,1\}^{6D} \mid x_0x_{2D}x_{4D}\in\{000,011,101,110,111\} \text{ and }x_j=0\text{ otherwise}\}$.  A general circuit with depth-$D$ and fan-in~$\leq 3$ that is geometrically restricted with respect to the cycle graph $C_n$ allows each output $y_j$ to depend on the input bits $\{x_{j\pm k} \mid k=1,\ldots,D\}$. Thus, upon feeding an input from $\mathcal{I}_{\mathrm{HLF},{6D}}^{(5)}$ to such a circuit, the only outputs that depend on more than one of the bits $x_0$, $x_{2D}$, or $x_{4D}$ are $y_{D} = y_{D}(x_0,x_{2D})$, $y_{3D} = y_{3D}(x_{2D},x_{4D})$, and $y_{5D} = y_{5D}(x_0,x_{4D})$.  Now let
\begin{subequations}
\begin{align}
    y^{(\mathrm{odd})}_{[0,2D]} &= \sum_{j=1}^{D} y_{2j-1} \\
    y^{(\mathrm{odd})}_{[2D,4D]} &= \sum_{j=D+1}^{2D} y_{2j-1} \\
    y^{(\mathrm{odd})}_{[4D,0]} &= \sum_{j=2D+1}^{3D} y_{2j-1} \\
    y^{(\mathrm{even})} & = \sum_{j=0}^{3D-1} y_{2j} .
\end{align}
\end{subequations}
Then we can parameterize each such function as
\begin{subequations}
\begin{align}
    y^{(\mathrm{odd})}_{[0,2D]} &= \alpha_D + \beta_D x_0 + \gamma_{D} x_{2D} + \delta_{D} x_0 x_{2D} \\
    y^{(\mathrm{odd})}_{[2D,4D]} &= \alpha_{3D} + \beta_{3D} x_{2D} + \gamma_{3D} x_{4D} + \delta_{3D} x_{2D} x_{4D} \\
    y^{(\mathrm{odd})}_{[4D,0]} &= \alpha_{5D} + \beta_{5D} x_{4D} + \gamma_{5D} x_{0} + \delta_{5D} x_0 x_{4D} \\
    y^{(\mathrm{even})} & = \alpha_0 + \beta_0 x_0 + \beta_{2D} x_{2D} + \beta_{4D} x_{4D}.
\end{align}
\end{subequations}
The win conditions for $\mathrm{SS}(C_{6D},\mathcal{I}_{\mathrm{HLF},{6D}}^{(5)})$ state that
\begin{subequations}
\begin{align}
    y^{(\mathrm{even})}(0,0,0) &= 0 \\
    y^{(\mathrm{even})}(1,1,0) + y^{(\mathrm{odd})}_{[0,2D]}(1,1) &= 1 \\
    y^{(\mathrm{even})}(0,1,1) + y^{(\mathrm{odd})}_{[2D,4D]}(1,1) &= 1 \\
    y^{(\mathrm{even})}(1,0,1) + y^{(\mathrm{odd})}_{[4D,0]}(1,1) &= 1 \\
    y^{(\mathrm{odd})}_{[0,2D]}(1,1) + y^{(\mathrm{odd})}_{[2D,4D]}(1,1) + y^{(\mathrm{odd})}_{[4D,0]}(1,1) &= 0.
\end{align}
\end{subequations}
However, substituting in the parameterization above and summing all the equations returns $0=1$.  Therefore, no circuit with depth-$D$ and fan-in~$\leq 3$ that is geometrically restricted with respect to the cycle graph $C_n$ can win the game $\mathrm{SS}(C_{6D},\mathcal{I}_{\mathrm{HLF},{6D}}^{(5)})$ with probability greater than $4/5$.

On the other hand, a perfect quantum strategy for this game can be performed by preparing the state $|C_n\rangle$ with a constant-depth quantum circuit and measuring each Pauli term in $E(\mathbf{1},\mathbf{x})$.  The measurement outcomes will deterministically satisfy Eq.~(\ref{eq:SS_win_conditions}) due to Lemma~\ref{lemma:Stab_Sub_Probs}.
\end{proof}

\begin{corollary}
Circuits in with depth-$(D+1)$ and fan-in~$\leq 3$ that are geometrically restricted with respect to the cycle graph $C_n$ can win the game $\mathrm{SS}(C_{6D},\mathcal{I}_{\mathrm{HLF},{6D}}^{(5)})$ on all inputs.
\end{corollary}
\begin{proof}
In the following all arithmetic is performed modulo 2.  Notice that in this case, $y_{D+1} = y_{D+1}(x_0,x_{2D})$, $y_{3D+1} = y_{3D+1}(x_{2D},x_{4D})$, and $y_{5D+1} = y_{5D+1}(x_{0},x_{4D})$.  Taking $y_{D+1} = x_0x_{2D}$, $y_{3D+1} = x_{2D}x_{4D}$, $y_{5D+1} = x_{0}x_{4D}$, and $y_j=0$ otherwise gives $y^{(\textrm{even})} = x_0x_{2D} + x_{2D}x_{4D} + x_{0}x_{4D}$ and $y^{(\textrm{odd})}_{[0,2D]} = y^{(\textrm{odd})}_{[2D,4D]} = y^{(\textrm{odd})}_{[4D,0]} = 0$, which satisfies all the above constraints imposed by the win conditions.
\end{proof}



\section{\label{sec:uncert_analysis}Uncertainty analysis}

In this appendix we provide details regarding the statistical estimates of our experiments. In particular, we make explicit the uncertainty contributions due to covariances between simultaneously measured stabilizers, and we show how the SPAM-correction procedure affects these error bars.

\subsection{Raw-data estimates}

In our tomography experiment, we prepare the $ n $-qubit state $ \rho $ using the graph-state formation circuit, further apply a local Clifford unitary $ U $ to rotate into the appropriate Pauli basis, and then perform a projective measurement in the computational basis $ \{ \ket{\mathbf{z}} \mid \mathbf{z} \in \{0,1\}^n \} $. Repeating this procedure $ N $ times yields a collection of $ N $ samples $ \mathbf{b}_1, \ldots, \mathbf{b}_N \in \{0,1\}^n $.

Let $ S \in \{I,X,Y,Z\}^{\otimes n} \times \{\pm 1\} $ be a (signed) $ n $-qubit Pauli observable diagonalized by $ U $, i.e.,
\begin{equation}\label{eq:diagonal_pauli}
	U S U^\dagger = \omega_S \prod_{j=1}^n Z^{f^{(S)}_j}_j,
\end{equation}
where $ f^{(S)}_j = 1 $ if $ S $ acts nontrivially on the $ j $th qubit, and $ 0 $ otherwise. Here, $ \omega_S \in \{\pm 1\} $ denotes the phase factor of $ S $. Then we can construct an estimator for $ \tr(\rho S) $ as
\begin{equation}
	\hat{\mu}_S = \frac{1}{N} \sum_{k=1}^N \omega_S (-1)^{\mathbf{f}^{(S)} \cdot \mathbf{b}_k},
\end{equation}
where $ \mathbf{f}^{(S)} \equiv (f^{(S)}_1, \ldots, f^{(S)}_n) $. Since $ S^2 = \openone $, an estimate for the variance is given by $ \hat{\sigma}_S^2 = 1 - \hat{\mu}_S^2 $, and in particular 
the standard error of the sample mean $ \hat{\mu}_S $ is
\begin{equation}
    \hat{\sigma}_{\hat{\mu}_S} = \frac{\hat{\sigma}_S}{\sqrt{N}}.
\end{equation}
This is the statistical uncertainty reported for each (non-SPAM-corrected) stabilizer expectation value in Table~\ref{tab:stabilizer_exp_values}.

To compute the standard error of our fidelity estimate, we need to address the covariances between stabilizer expectation values that are evaluated within the same measurement setting. Given a set $ \mathcal{C} $ of locally commuting Pauli operators, referred to as a clique, there exists some local Clifford transformation $ U = U_\mathcal{C} $ such that Eq.~\eqref{eq:diagonal_pauli} holds for all $ S \in \mathcal{C} $. Then we can evaluate $ \hat{\mu}_S $ for each $ S \in \mathcal{C} $, using the same samples $ \mathbf{b}_k $, as prescribed above.

Now consider
\begin{equation}
    O_\mathcal{C} = \sum_{S \in \mathcal{C}} h_S S, \quad h_S \in \R.
\end{equation}
By linearity, $ \tr(\rho O_\mathcal{C}) $ is estimated via
\begin{equation}
    \hat{\mu}_{O_\mathcal{C}} = \sum_{S \in \mathcal{C}} h_S \hat{\mu}_S,
\end{equation}
which has variance
\begin{equation}
    \begin{split}
    \V[\hat{\mu}_{O_\mathcal{C}}] &= \Cov[\hat{\mu}_{O_\mathcal{C}}, \hat{\mu}_{O_\mathcal{C}}]\\
    &= \sum_{S, T \in \mathcal{C}} h_S h_T \Cov[\hat{\mu}_S, \hat{\mu}_T].
    \end{split}
\end{equation}
We already have an estimate for the diagonal terms, $ \Cov[\hat{\mu}_S, \hat{\mu}_S] = \V[\hat{\mu}_S] $, via $ \sigma_S^2 $. The off-diagonal covariances $ \Cov[\hat{\mu}_S, \hat{\mu}_T] $ are estimated as
\begin{equation}\label{eq:covariance_estimate}
    \hat{\Sigma}_{S, T} = \hat{\mu}_{ST} - \hat{\mu}_S \hat{\mu}_T.
\end{equation}
Note that this expression involves the Pauli operator $ ST $, even if $ ST \notin \mathcal{C} $. Nonetheless, since $ S $ and $ T $ locally commute, the samples obtained from this measurement setting contain sufficient information to compute $ \hat{\mu}_{ST} $. Recognizing that $ \hat{\Sigma}_{S,S} = \hat{\sigma}_S^2 $, we have a compact expression for the variance of $ \hat{\mu}_{O_\mathcal{C}} $,
\begin{equation}
    \hat{\sigma}_{O_\mathcal{C}}^2 = \sum_{S, T \in \mathcal{C}} h_S h_T \hat{\Sigma}_{S, T}.
\end{equation}

Finally, suppose we have disjoint cliques $ \mathcal{C}_1, \ldots, \mathcal{C}_L $, corresponding to the different measurement settings of our experiment. Such cliques are constructed in order to estimate the observable $ O = \sum_{\ell=1}^L O_{\mathcal{C}_\ell} $. (For instance, our fidelity experiment featured $ O = \op{C_6}{C_6} $, $ L = 37 $, and $ h_S = 2^{-6} \ \forall S \in \mathcal{S}_{C_6} $.) Then we set
\begin{equation}
    \hat{\mu}_O = \sum_{\ell=1}^L \sum_{S \in \mathcal{C}_\ell} h_S \hat{\mu}_S,
\end{equation}
which has standard error
\begin{equation}
    \hat{\sigma}_{\hat{\mu}_O} = \sqrt{\frac{1}{N} \sum_{\ell=1}^L \sum_{S, T \in \mathcal{C}_\ell} h_S h_T \hat{\Sigma}_{S,T}}.
\end{equation}
For simplicity, we have assumed that all cliques were sampled the same number of times $ N $, as was the case in our experiment. This is the statistical uncertainty for the fidelity estimate reported in the main text. Note that there are no covariances between estimators of different cliques, since they correspond to separate, independently obtained samples.

\subsection{SPAM-corrected estimates}

To correct for SPAM errors in our device, we ran a series of calibrating experiments to characterize measurement-readout error rates. This corresponds to constructing a $ 2^n \times 2^n $ stochastic matrix $ M $ such that, for the probability distribution $ q = (q_\mathbf{z} \mid \mathbf{z} \in \{0,1\}^n)^\T $ encoded in the prepared quantum state, we instead observe (due to measurement errors solely) the distribution $ p = Mq $. By computing the inverse matrix $ M^{-1} $, we can recover the SPAM-error-free distribution $ q = M^{-1}p $.

In practice, we only have access to the estimate $ \hat{p} $ of $ p $, constructed from the samples $ \mathbf{b}_1, \ldots, \mathbf{b}_N $. Applying this inversion procedure to $ \hat{p} $, we obtain
\begin{align}
    \hat{q}_z &= \sum_{\mathbf{z}' \in \{0,1\}^n} [M^{-1}]_{\mathbf{z},\mathbf{z}'} \, \hat{p}_{\mathbf{z}'} \notag\\
    &= \sum_{\mathbf{z}' \in \{0,1\}^n} \l( [M^{-1}]_{\mathbf{z},\mathbf{z}'} \frac{1}{N} \sum_{k=1}^N \delta_{\mathbf{z}',\mathbf{b}_k} \r) \notag\\
    &= \frac{1}{N} \sum_{k=1}^N [M^{-1}]_{\mathbf{z},\mathbf{b}_k}.
\end{align}
All SPAM-corrected expectation values (and their associated uncertainties) are then evaluated by simply replacing $ \hat{p} $ with $ \hat{q} $. For instance, in terms of $ M^{-1} $, the expectation-value estimates become
\begin{equation}\label{eq:spam_ev}
    \hat{\mu}_S = \frac{1}{N} \sum_{k=1}^N \l[ \sum_{\mathbf{z} \in \{0,1\}^n} [M^{-1}]_{\mathbf{z},\mathbf{b}_k} \omega_S (-1)^{\mathbf{f}^{(S)} \cdot \mathbf{z}} \r].
\end{equation}
Similarly, the SPAM-corrected covariance matrix between clique elements is
\begin{widetext}
\begin{equation}
    \hat{\Sigma}_{S, T} = \frac{1}{N} \sum_{\mathbf{z}, \mathbf{z}' \in \{0,1\}^n} \l[ \l( \sum_{k=1}^N [M^{-1}]_{\mathbf{z},\mathbf{b}_i} [M^{-1}]_{\mathbf{z}',\mathbf{b}_i} \r) \omega_S \omega_T (-1)^{\mathbf{f}^{(S)} \cdot \mathbf{z} + \mathbf{f}^{(T)} \cdot \mathbf{z}'} \r] - \hat{\mu}_S \hat{\mu}_T,
\end{equation}
\end{widetext}
where $ \hat{\mu}_S, \hat{\mu}_T $ are SPAM-corrected estimates, per Eq.~\eqref{eq:spam_ev}.



\section{\label{sec:additional_data}Additional experimental data}

In Table~\ref{tab:stabilizer_exp_values} we explicitly list our experimental estimates for each stabilizer expectation value, before and after SPAM correction, along with $ 1\sigma $ statistical uncertainties (see Appendix~\ref{sec:uncert_analysis}). These values are the same as those plotted in Fig.~\ref{fig:stabilizer_values_bar_graph} of the main text. Furthermore, in this table the stabilizers are partitioned into our choice of locally commuting sets, such that we measure all stabilizers in each set by a single Pauli measurement setting.

\setlength\LTcapwidth{\linewidth}
\begin{longtable}{c @{\extracolsep{\fill}} c c c}
    
    \caption{Expectation values for the 63 nontrivial stabilizers of $ \ket{C_6} $ estimated from the CBF experiment. Uncertainties denote $1\sigma$ standard errors (details in Appendix~\ref{sec:uncert_analysis}). Horizontal rules indicate our partitioning of the stabilizers into 37 locally commuting sets, such that we estimate all stabilizers in that set from the same global Pauli measurement.}
    \label{tab:stabilizer_exp_values}
    \\
    \toprule
    {} & {} & \multicolumn{2}{c}{Expectation value} \\
    \cline{3-4}
    Input & Stabilizer & Raw value & SPAM-corrected \\
    \hline
    \endfirsthead
    
    \caption*{TABLE~\ref*{tab:stabilizer_exp_values}.~\emph{(Continued)}.}
    \\
    \toprule
    {} & {} & \multicolumn{2}{c}{Expectation value} \\
    \cline{3-4}
    Input & Stabilizer & Raw value & SPAM-corrected \\
    \hline
    \endhead
    
    \hline
    \endfoot
    
    \botrule
    \endlastfoot

$ 000001 $ & $ +ZIIIZX $ & $ 0.7164(10) $ & $ 0.7616(10) $ \\
$ 001000 $ & $ +IZXZII $ & $ 0.6732(10) $ & $ 0.7178(11) $ \\
$ 001001 $ & $ +ZZXZZX $ & $ 0.4952(12) $ & $ 0.5627(14) $ \\
\hline
$ 000010 $ & $ +IIIZXZ $ & $ 0.7208(10) $ & $ 0.7654(10) $ \\
$ 010000 $ & $ +ZXZIII $ & $ 0.7320(10) $ & $ 0.7787(10) $ \\
$ 010010 $ & $ +ZXZZXZ $ & $ 0.5352(12) $ & $ 0.6046(14) $ \\
\hline
$ 000011 $ & $ +ZIIZYY $ & $ 0.7656(9) $ & $ 0.8302(10) $ \\
$ 010011 $ & $ +IXZZYY $ & $ 0.6168(11) $ & $ 0.6813(12) $ \\
\hline
$ 000100 $ & $ +IIZXZI $ & $ 0.6764(10) $ & $ 0.7205(11) $ \\
$ 100000 $ & $ +XZIIIZ $ & $ 0.7200(10) $ & $ 0.7646(10) $ \\
$ 100100 $ & $ +XZZXZZ $ & $ 0.5324(12) $ & $ 0.6019(14) $ \\
\hline
$ 000101 $ & $ +ZIZXIX $ & $ 0.5748(12) $ & $ 0.6221(13) $ \\
$ 010001 $ & $ +IXZIZX $ & $ 0.5584(12) $ & $ 0.6034(13) $ \\
$ 010100 $ & $ +ZXIXZI $ & $ 0.5900(11) $ & $ 0.6448(12) $ \\
\hline
$ 000110 $ & $ +IIZYYZ $ & $ 0.5360(12) $ & $ 0.5803(13) $ \\
$ 100110 $ & $ +XZZYYI $ & $ 0.4368(13) $ & $ 0.4862(14) $ \\
\hline
$ 000111 $ & $ -ZIZYXY $ & $ 0.6112(11) $ & $ 0.6765(12) $ \\
$ 010111 $ & $ -IXIYXY $ & $ 0.5936(11) $ & $ 0.6449(12) $ \\
\hline
$ 001010 $ & $ +IZXIXZ $ & $ 0.5536(12) $ & $ 0.5984(13) $ \\
$ 100010 $ & $ +XZIZXI $ & $ 0.5392(12) $ & $ 0.5884(13) $ \\
$ 101000 $ & $ +XIXZIZ $ & $ 0.5832(11) $ & $ 0.6311(12) $ \\
\hline
$ 001011 $ & $ +ZZXIYY $ & $ 0.6084(11) $ & $ 0.6718(12) $ \\
\hline
$ 001100 $ & $ +IZYYZI $ & $ 0.7376(10) $ & $ 0.8037(10) $ \\
$ 101100 $ & $ +XIYYZZ $ & $ 0.6316(11) $ & $ 0.6981(12) $ \\
\hline
$ 001101 $ & $ +ZZYYIX $ & $ 0.6068(11) $ & $ 0.6712(12) $ \\
\hline
$ 001110 $ & $ -IZYXYZ $ & $ 0.5924(11) $ & $ 0.6565(13) $ \\
$ 101110 $ & $ -XIYXYI $ & $ 0.5308(12) $ & $ 0.5784(13) $ \\
\hline
$ 001111 $ & $ +ZZYXXY $ & $ 0.7056(10) $ & $ 0.7970(11) $ \\
\hline
$ 010101 $ & $ +IXIXIX $ & $ 0.5368(12) $ & $ 0.5690(13) $ \\
$ 111110 $ & $ -YXXXYI $ & $ 0.6228(11) $ & $ 0.6925(12) $ \\
\hline
$ 010110 $ & $ +ZXIYYZ $ & $ 0.4936(12) $ & $ 0.5483(14) $ \\
\hline
$ 011000 $ & $ +ZYYZII $ & $ 0.5768(12) $ & $ 0.6275(13) $ \\
$ 011010 $ & $ +ZYYIXZ $ & $ 0.4676(13) $ & $ 0.5171(14) $ \\
\hline
$ 011001 $ & $ +IYYZZX $ & $ 0.4500(13) $ & $ 0.4976(14) $ \\
\hline
$ 011011 $ & $ +IYYIYY $ & $ 0.5200(12) $ & $ 0.5628(13) $ \\
$ 101101 $ & $ +YIYYIY $ & $ 0.5584(12) $ & $ 0.6037(13) $ \\
$ 110110 $ & $ +YYIYYI $ & $ 0.5016(12) $ & $ 0.5487(13) $ \\
\hline
$ 011100 $ & $ -ZYXYZI $ & $ 0.6444(11) $ & $ 0.7192(12) $ \\
$ 011101 $ & $ -IYXYIX $ & $ 0.5524(12) $ & $ 0.5982(13) $ \\
\hline
$ 011110 $ & $ +ZYXXYZ $ & $ 0.5528(12) $ & $ 0.6233(13) $ \\
\hline
$ 011111 $ & $ -IYXXXY $ & $ 0.6544(11) $ & $ 0.7232(12) $ \\
\hline
$ 100001 $ & $ +YZIIZY $ & $ 0.5404(12) $ & $ 0.5872(13) $ \\
$ 101001 $ & $ +YIXZZY $ & $ 0.4708(12) $ & $ 0.5214(14) $ \\
\hline
$ 100011 $ & $ -YZIZYX $ & $ 0.6172(11) $ & $ 0.6856(12) $ \\
$ 101011 $ & $ -YIXIYX $ & $ 0.5836(11) $ & $ 0.6308(12) $ \\
\hline
$ 100101 $ & $ +YZZXIY $ & $ 0.4624(13) $ & $ 0.5114(14) $ \\
\hline
$ 100111 $ & $ +YZZYXX $ & $ 0.5476(12) $ & $ 0.6190(13) $ \\
\hline
$ 101010 $ & $ +XIXIXI $ & $ 0.5500(12) $ & $ 0.5863(13) $ \\
$ 111101 $ & $ -XXXYIY $ & $ 0.6764(10) $ & $ 0.7486(12) $ \\
\hline
$ 101111 $ & $ -YIYXXX $ & $ 0.6884(10) $ & $ 0.7606(11) $ \\
\hline
$ 110000 $ & $ +YYZIIZ $ & $ 0.7484(9) $ & $ 0.8090(10) $ \\
$ 110100 $ & $ +YYIXZZ $ & $ 0.5948(11) $ & $ 0.6597(13) $ \\
\hline
$ 110001 $ & $ -XYZIZY $ & $ 0.6144(11) $ & $ 0.6785(12) $ \\
$ 110101 $ & $ -XYIXIY $ & $ 0.5580(12) $ & $ 0.6048(13) $ \\
\hline
$ 110010 $ & $ +YYZZXI $ & $ 0.5880(11) $ & $ 0.6549(13) $ \\
\hline
$ 110011 $ & $ +XYZZYX $ & $ 0.6952(10) $ & $ 0.7871(11) $ \\
\hline
$ 110111 $ & $ -XYIYXX $ & $ 0.6428(11) $ & $ 0.7138(12) $ \\
\hline
$ 111000 $ & $ -YXYZIZ $ & $ 0.6248(11) $ & $ 0.6893(12) $ \\
$ 111010 $ & $ -YXYIXI $ & $ 0.5592(12) $ & $ 0.6091(13) $ \\
\hline
$ 111001 $ & $ +XXYZZY $ & $ 0.5164(12) $ & $ 0.5832(14) $ \\
\hline
$ 111011 $ & $ -XXYIYX $ & $ 0.6248(11) $ & $ 0.6908(12) $ \\
\hline
$ 111100 $ & $ +YXXYZZ $ & $ 0.7492(9) $ & $ 0.8456(11) $ \\
\hline
$ 111111 $ & $ +XXXXXX $ & $ 0.8304(8) $ & $ 0.9378(9) $
\end{longtable}


\bibliography{references}

\end{document}